\documentclass[a4paper,USenglish,cleveref,showAuthorOrcid,hideChapterNumber]{lipics-v2021}
\usepackage{craigstyle}
\usepackage[utf8]{inputenc}
\usepackage[T1]{fontenc}

\usepackage{etoolbox}
\usepackage[noadjust]{cite}

\usepackage{amsmath,amssymb}

\usepackage{graphicx}

\usepackage{relsize}

\usepackage{colonequals}

\usepackage{adjustbox}

\usepackage{multirow}

\usepackage[dvipsnames]{xcolor}

\usepackage[all]{xy}

\usepackage[noxy]{virginialake}

\usepackage{macros_Iris} 

\usepackage{tikz}
\usetikzlibrary{positioning,arrows,calc,shapes.geometric}
\tikzset{ modal/.style={>=stealth,shorten >=1pt,shorten <=1pt,auto,node distance=0.7cm, semithick}, world/.style={circle,draw,minimum size=0.2cm}, point/.style={circle,draw,inner sep=0.5mm,fill=black},point1/.style={circle,draw,inner sep=0.5mm}, reflexive above/.style={->,loop,looseness=7,in=120,out=60}, reflexive below/.style={->,loop,looseness=7,in=240,out=300}, reflexive left/.style={->,loop,looseness=7,in=150,out=210}, reflexive right/.style={->,loop,looseness=7,in=30,out=330}, itria/.style={draw,isosceles triangle,shape border rotate=270,yshift=-1.45cm, minimum height = 15mm}, empty/.style={inner sep=0.0mm } }

\usepackage{url}

\usepackage{hyperref}

\pdfoutput=1

\newcommand{\posint}{\mathbb{N}}

\newcommand{\relat}{\mathcal{R}}

\catcode`@=11
\def\Left#1#2\Right{\begingroup%
   \def\ts@r{\nulldelimiterspace=0pt \mathsurround=0pt}%
   \let\@hat=#1%
   \def\sht@im{#2}%
   \def\@t{{\mathchoice{\def\@fen{\displaystyle}\k@fel}%
          {\def\@fen{\textstyle}\k@fel}%
          {\def\@fen{\scriptstyle}\k@fel}%
          {\def\@fen{\scriptscriptstyle}\k@fel}}}%
   \def\g@rin{\ts@r\left\@hat\vphantom{\sht@im}\right.}%
   \def\k@fel{\setbox0=\hbox{$\@fen\g@rin$}\hbox{%
      $\@fen \kern.4\wd0 \copy0 \kern-.4\wd0%
      \llap{\copy0}\kern.4\wd0$}}%
      \def\pt@h{\mathopen\@t}\pt@h\sht@im%
      \Right}%
\def\Right#1{\let\@hat=#1%
   \def\st@m{\mathclose\@t}%
   \st@m\endgroup}
\catcode`@=12

\newcommand{\Lab}{\mathtt{L}}

\newcommand{\intfun}{\seqint}
\newcommand{\multform}{\mathop{\mathsf{form}}}

\newcommand{\class}{\mathcal{C}}

\renewcommand{\emptyset}{\varnothing}
\renewcommand{\phi}{\varphi}
\newcommand{\impl}{\rightarrow}
\renewcommand{\epsilon}{\varepsilon}
\newcommand{\IPC}{{\sf IPC}}
\newcommand{\K}{\textnormal{{\sf K}}}
\newcommand{\D}{\textnormal{{\sf D}}}
\newcommand{\T}{\textnormal{{\sf T}}}
\newcommand{\GL}{\textnormal{{\sf GL}}}

\newcommand{\vlfill}[1]{\{#1\}}
\renewcommand{\vlhole}{\vlfill{\:}}

\newcommand{\ce}{\colonequals}
\newcommand{\cce}{\coloncolonequals}

\newcommand{\calM}{\mathcal{M}}

\newcommand{\calL}{\mathcal{L}}
\newcommand{\calLp}{\calL_{\sf{p}}}
\newcommand{\calLBox}{\calL_{\Box}}

\newcommand{\spdisj}{\mathbin{\varovee}}
\newcommand{\spconj}{\mathbin{\varowedge}}
\newcommand{\bigspdisj}{\mathop{\mathlarger{\mathlarger{\mathlarger{\mathlarger{\varovee}}}}}\limits}
\newcommand{\bigspconj}{\mathop{\mathlarger{\mathlarger{\mathlarger{\mathlarger{\varowedge}}}}}\limits}
\newcommand{\form}{{\sf Form}}

\newcommand{\seqint}{\Left[\cdot\Right]}
\newcommand{\labint}[1]{\Left[#1\Right]}

\newcommand{\itRR}{\mathit{RR}}
\newcommand{\itRL}{\mathit{RL}}
\newcommand{\itLR}{\mathit{LR}}
\newcommand{\itLL}{\mathit{LL}}

\newcommand{\Icomment}[1]{ }

\newcommand{\labe}[2]{{#1}\mathord{:}{#2}}

\title{Interpolation in Proof Theory}
\author{Iris van der Giessen}{University of Amsterdam, The Netherlands}{i.vandergiessen@uva.nl}{https://orcid.org/0009-0008-4908-2496}{} 
\author{Raheleh Jalali}{University of Bath, UK}{rahele.jalali@gmail.com}{https://orcid.org/0000-0002-3321-8087}{}
\author{Roman Kuznets}{Institute of Computer Science of the Czech Academy of Sciences, Czechia}{kuznets@cs.cas.cz}{https://orcid.org/0000-0001-5894-8724}{}

\begin{document}

\maketitle              % typeset the header of the contribution

\begin{abstract}
  This chapter provides a comprehensive overview of proof-theoretic methods for establishing interpolation properties   
  across a range of logics, including classical, intuitionistic, modal, and substructural logics. Central to the discussion are two foundational techniques: Maehara's method for Craig interpolation
  and Pitts' method for uniform interpolation. The chapter demonstrates how these methods lead to results on the existence of well-behaved proof systems in the contemporary framework of universal proof theory and how they provide a road map for constructing interpolation proofs using modern proof formalisms. The emphasis of the chapter is on constructive, modular, and syntax-driven techniques that illuminate deeper connections between interpolation properties and proof systems. 
\end{abstract}

\setcounter{tocdepth}{2}
 \tableofcontents

\section{Introduction}  

Proof-theoretic techniques  towards interpolation originating from the 1950/60's~\cite{Craig57a,Maehara61} are well-established  and play an important role in the contemporary research on interpolation. The aim of this chapter is to provide a nice balance between (1)~the~\emph{basics} of the praised proof-theoretic methods based on sequent calculi to prove interpolation properties~(IPs), (2)~the~recent \emph{applications} of such methods to the existence of proof systems within the emerging field of \emph{universal proof theory}, and (3)~the~\emph{development} of similar methods for more powerful contemporary proof systems such as labelled sequents.

The proof-theoretic method to prove the Craig interpolation property~(CIP) is known as Maehara's method~\cite{Maehara61}. It provides an algorithm that, given a sequent proof of $\phi \impl \psi$,  produces a Craig interpolant~$\theta$ for it. This method is \emph{constructive}: it explicitly constructs an interpolant, instead of merely proving its existence. Sometimes, the constructed interpolants preserve the polarity of propositional variables, which then shows the so-called Lyndon interpolation property~(LIP). Maehara's method is used for a wide range of logics, including classical and intuitionistic first-order logic~\cite{Takeuti75,TroelstraS00,Schuette77}, (intuitionistic) modal logic~\cite{Fitting83,vdGiessen22}, linear logic~\cite{Roorda94}, and substructural logic~\cite{Onokomori}.

Pitts~\cite{Pitts92} has developed a syntactic method to prove the uniform interpolation property~(UIP) for the intuitionistic propositional logic~$\IPC$ (see  \refchapter{chapter:uniform},  the main chapter on uniform interpolation). At the time, this was a rather surprising result as it shows that propositional quantification can be interpreted in~$\IPC$ itself. Pitts uses a strongly terminating sequent calculus for~$\IPC$ to prove the~UIP. Pitts' idea was generalized to various propositional and (intuitionistic) modal logics~\cite{Bilkova06,Iemhoff19APAL,Iemhoff19AML,AkbarTabatabaiIJ22nonnormal,Feree_etal24} and  formalized in Coq/Rocq~\cite{Feree_vGool23,Feree_etal24} enabling the automatic computation of uniform interpolants.

Maehara's method and Pitts' method are based on sequent calculi. With the development of more expressive sequent structures, such as labelled sequents, hypersequents, and nested sequents, related methods were designed to prove interpolation properties using such structures~\cite{KuznetsL18,Kuznets16LFCS,Lyon_etall20,vdGiessenJK24,vdGiessenJK23TABLEAUX}. The new methods reveal that a more robust type of `interpolants' is needed where it does not suffice to define them as formulas, but rather as structures that mimic the type of sequent structure at hand. Interestingly, such calculi seem more suitable to prove the~LIP than Maehara's method via analytic sequent calculi. 

The proof-theoretic approach to interpolation is not only useful in light of its constructive nature. The approach can also establish general results in proof theory itself.  In the field of \emph{universal proof theory}, structural conditions on sequent rules are identified that guarantee the~CIP~or~UIP~\cite{Iemhoff19APAL,AkbarTabatabaiJ25,AkbarTabatabaiJ18uniform,AkbarTabatabaiIJ22nonnormal}. Such sequent rules are called \emph{semi-analytic} and capture the intuitive notion of `nice' 
rules. Given the fact that interpolation properties are rare among intermediate logics, most such logics cannot have such `nice' sequent calculi. 

The advantages of the proof-theoretic method are the following (we already mentioned a couple): it provides a \emph{constructive} proof of interpolation properties; the complexity of a Craig interpolant is often \emph{linear} in the size of the proof; it is flexible across different classes of logics; and it finds applications in \emph{proof complexity} and \emph{universal proof theory}. There are also some disadvantages. 
Although the method is \emph{constructive}, it does not immediately provide an algorithm for computing
an interpolant based on a provable implication $\phi \impl \psi$ because it requires the construction of a sequent proof first. Also, Maehara's method is not complete meaning that it might be that given $\phi \impl \psi$ some of its interpolants cannot be constructed via this method~\cite{HetzlJalali24}. \looseness=-1

Some applications and developments of the proof-theoretic approach to interpolation can be found in other chapters of this book.  \refchapter{chapter:proofcomplexity} on \emph{proof complexity} discusses the feasible extraction of interpolants from proofs, which relates to the NP~vs.~coNP problem.  \refchapter{chapter:fixedpoint} discusses new Maehara-type methods for \emph{cyclic proofs}.

This chapter is structured as follows. 
We choose to illustrate the methods for the simple propositional and modal logics introduced in \cref{subseq:prel}. The basics of Maehara's method and Pitts' method are explained in \cref{sec:basics}. The modularization of these methods and its impact on universal proof theory is discussed in \cref{sec:universal-proof-theory}. In \cref{sect:gener} the proof-theoretic method and results for labelled sequents, hypersequents, and nested sequents is discussed. In \cref{sec:conclusion} we provide some concluding remarks.

\section{Preliminaries}
\label{subseq:prel}

\subsection{Propositional and Modal Logic}

Fix a countable set of \emph{propositional variables} $\Prop = \{p,q,\dots\}$, also called \emph{atoms}. 
For propositional logic we work with the language $\calLp \ce \{ \bot, \top, \vee, \wedge, \impl \}$  and use symbols $\phi, \psi, \dots$ to denote formulas generated from~$\Prop$~and~$\calLp$ as usual. 
The set of all formulas is denoted by~$\form_{\calLp}$. 
We use the standard abbreviation $\neg \phi \ce  \phi \impl \bot$. 
The base logics used in this chapter are \textit{classical propositional logic}~$\CPC$ and \textit{intuitionistic propositional logic}~$\IPC$, where $\IPC$~can be defined as $\CPC$~without the law of the excluded middle $\phi \vee \neg \phi$. 
An \emph{intermediate logic} is any logic~$\lgc$ such that $\IPC \subseteq \lgc \subseteq \CPC$, where we consider a logic as a set of formulas closed under substitution and modus ponens. 
The interested reader is referred to \cite{vDalenTroelstra88,TroelstraS00}.

For modal logic we work with the language $\calLBox \ce \calLp \cup \{ \Box \}$. 
As usual we define $\Diamond \phi \ce  \neg \Box \neg \phi$.  
We define~$\Box^n \phi$ inductively as $\Box^0 \phi \ce \phi$ and $\Box^{n+1} \phi \ce \Box (\Box^n \phi)$.  \Cref{fig:modal-axioms} presents standard modal axioms and rules, and the left column of \cref{fig:modal-logics} presents common modal logics. 
Semantically, modal logics can be characterized by Kripke frames, also known as  \textit{linear transition systems}. 
A \emph{Kripke model} is a tuple $\calM = (W,R,V)$ with a non-empty set~$W$ of \emph{worlds}, a binary relation $R \subseteq W \times W$ called the \emph{accessibility relation}, and a function \mbox{$V \colon W \to \mathcal{P}(\Prop)$} called a \emph{valuation}.  
\Cref{fig:modal_satisfaction} provides the inductive definition of the \emph{modal satisfaction relation}~$\models$ between worlds in a model and formulas. 
We say that formula~$\phi$ is \emph{valid} in model $\calM = (W,R,V)$ if{f} $\calM,w \models \phi$ for all worlds $w \in W$.  
Well-known soundness and completeness results are summarized in \cref{fig:modal-logics}, where $\phi$~is in a logic whenever it is valid in all respective models of the indicated classes. 
We refer to~\cite{Blackburnetall01ML,ChagrovZakharyaschev97ModalLogic} for full introductions on modal logic. 

\begin{figure}[t]
  \centering
\fbox{\parbox{0.97\textwidth}{
{\centering
$
\begin{array}{ll}
    \kaxiom &\Box (\phi \impl \psi) \impl (\Box \phi \impl \Box \psi)\\
   \taxiom &\Box \phi \impl \phi\\
    \fouraxiom &\Box \phi \impl \Box \Box \phi \\
\end{array}
\quad
\begin{array}{ll}
    \fiveaxiom &\neg \Box \phi \impl \Box \neg \Box \phi\\
    \daxiom &\Box \phi \impl \Diamond \phi\\
    \wlobaxiom & \Box (\Box \phi \impl \phi) \impl \Box \phi
\end{array}
\quad
\begin{array}{ll}
    {\MP} & \vliinf{}{}{\psi}{\phi \impl \psi}{\phi}\\    
    {\Nrule} & \vlinf{}{}{\Box \phi}{\phi}
\end{array}
$
}}}
\caption{Modal axioms and the inference rules  modus ponens~{\MP} and necessitation~{\Nrule}.} 
\label{fig:modal-axioms}
\end{figure}

\begin{figure}[t!]
    \centering
    \fbox{\parbox{0.98\textwidth}{

    $
    \begin{array}{lll}
	 \calM, w \models p &\text{if{f} } & p \in V(w);\\
    \calM, w \models\top; & &\\
    \calM, w\not\models\bot; & & \\
    \calM, w\models\phi\impl\psi &\text{if{f} } & \calM, w\not\models\phi \text{ or }\calM, w\models\psi;\\
    \calM, w\models\phi\vee\psi &\text{if{f} } & \calM, w\models\phi \text{ or }\calM, w\models\psi;\\
    \calM, w\models\phi\wedge\psi &\text{if{f} } & \calM, w\models\phi \text{ and }\calM, w\models\psi;\\
    \calM, w\models\Box \phi &\text{if{f} } & \text{for all } v \text{ such that } wRv, \text{ we have }\calM, v\models\phi.
    \end{array}
    $
    }}
    \caption{Definition of the modal satisfaction relation~$\models$ between Kripke model $\calM = (W,R,V)$, world $w \in W$, and formula~$\phi$.}
\label{fig:modal_satisfaction}
\end{figure}

\begin{figure}[t!]
    \small{
\fbox{\parbox{0.97\textwidth}{
\setlength\abovedisplayskip{0pt}
\setlength\belowdisplayskip{0pt}
\begin{align*}
    \K &\ce  \CPC  + \kaxiom   && \text{class of all models}\\
    \T &\ce   \K  + \taxiom    && \text{class of all reflexive models}\\
    \D &\ce  \K + \daxiom     && \text{class of all serial models}\\
    \Kvier &\ce  \K + \fouraxiom    && \text{class of all transitive models}\\
    \Svier &\ce  \K + \taxiom  + \fouraxiom    && \text{class of all transitive and reflexive models}\\
    \Svijf &\ce  \K + \taxiom  + \fiveaxiom &&\text{class of all total models}\\
    \GL &\ce  \K + \wlobaxiom && \text{class of all transitive and conversely well-founded models}
\end{align*}
}
}}
\caption{Axiomatizations of classical modal logics (each should be understood as closed under rules~{\MP}~and~{\Nrule}) with their soundness and completeness results.}
\label{fig:modal-logics}
\end{figure}

In general, a language~$\calL$ is a set of symbols and we write~$\form_{\calL}$ to denote the set of formulas generated from~$\Prop$~and~$\calL$. For every language, we use symbols $\phi, \psi, \dots$ to denote formulas. We write $\Sub{\phi}$ for the set of all subformulas occurring in formula~$\phi$. For a logic~$\lgc$, we write $\Lproves \phi$ to mean $\phi \in \lgc$.

\subsection{Sequent Calculi}
\label{subsec:seqcal}

Sequent calculi are rule-based descriptions of a logic and originate from Gentzen~\cite{Gentzen35a,Gentzen35b}. Here we introduce sequent calculi for propositional logic and modal logic~\cite{Gentzen35a,Takeuti75,Schuette77,TroelstraS00}. 

A \emph{sequent} is an expression of the form $\Gamma\seqar\Delta$ where $\Gamma$~and~$\Delta$~are finite multisets of formulas. 
$\Gamma$~is called the \emph{antecedent} of such a sequent, and $\Delta$~is called its \emph{succedent} or \emph{conclusion}. If $\Delta$~contains at most one formula, the sequent is called a \emph{single-conclusion} sequent.  We use capital letters $S, S', S_1, S_2, \dots$ to refer to sequents. We use the notation $\Gamma, \varphi$ for the multiset $\Gamma \cup \{\varphi\}$ and $\Gamma, \Delta$ for $\Gamma \cup \Delta$. We define $\Sub{\Gamma} \ce  \bigcup_{\phi \in \Gamma} \Sub{\phi}$ for the set of all subformulas of formulas occurring in~$\Gamma$ and $\Box \Gamma \ce  \{\Box \phi \mid \phi \in \Gamma \}$.  

The intended meaning of a sequent $\Gamma\seqar\Delta$ is the formula 
$
\bigwedge_{\phi \in \Gamma}\phi\impl\bigvee_{\psi \in \Delta}\psi
$
also called the \emph{formula interpretation} of the sequent. We adopt the convention that the empty conjunction is~$\top$ and the empty disjunction is~$\bot$. 

An \emph{inference rule} is an expression of the form on the left below where $S$~and~$S_i$'s~are sequents. Each inference rule gives rise to different \emph{rule instances}, namely for every substitution~$\sigma$, the expression on the right is an \emph{instance} of the rule.
\begin{equation}
\label{eq:seq_rule}
\vliiiinf{}{}{S}{S_1}{S_2}{\dots}{S_n}  \qquad \qquad \qquad
\vliiiinf{}{}{\sigma{S}}{\sigma{S_1}}{\sigma{S_2}}{\dots}{\sigma{S_n}}
\end{equation}
This is a rather formal way of presenting sequent rules needed in \cref{sec:semi-analytic}, 
but for now we adopt a standard practice in proof theory of leaving the substitutions implicit and write rules as \emph{rule schemas} (e.g.,~\cref{fig:G3pc}), in which so-called \emph{multiset variables} (resp.~formula variables and atom variables) in the rules can represent any multiset (resp.~formula and atom).

The upper sequent(s) of a rule are called the \emph{premise(s)} and the lower sequent is called the \emph{conclusion}. A rule with no premises is called an \emph{axiom}, a rule with one premise is a \emph{unary} rule, and a rule with two premises is a \emph{binary} rule. When considering rule schemas as in \cref{fig:G3pc}, the distinguished formula in the conclusion is called the \emph{principal formula} and the multisets like $\Gamma, \Delta$, etc.~are called the \emph{context}. The formulas in the premises that are not in the context are called the \emph{active} formulas of the rule. Axiom~$(\id)$ differs in that it has two principal formulas but no active ones. In the cut rule (\cref{fig:G3pc}) we call~$\phi$ the \emph{cut formula}.

A \emph{sequent calculus} is a finite set of inference rules and a \emph{proof} in a given sequent calculus~$\SC$ is a finite rooted tree in which the nodes are sequents generated inductively from instances of the inference rules in~$\SC$. This means that in a proof all leaves are instances of axioms. The root of the tree is called the \emph{endsequent} and is the sequent proved by the sequent proof.  
We write $\SCproves  \Gamma \seqar \Delta$ if{f} there is a proof in sequent calculus~$\SC$ with endsequent $\Gamma \seqar \Delta$.

Sequent calculi are used to represent the provable formulas of a logic via soundness and completeness theorems. This is reflected in the following definition:

\begin{definition}
\label{CalculusForLogic}
    Let $\SC$~be a sequent calculus and let $\lgc$~be a logic over the same language~$\calL$. We say that $\SC$~is a \emph{sequent calculus for}~$\lgc$, or that~$\lgc$ is \emph{the logic of}~$\SC$, if{f} for every sequent $\Gamma \seqar \Delta$, we have $\SCproves  \Gamma \seqar \Delta$ if{f} $\Lproves \bigwedge \Gamma \impl \bigvee \Delta$. 
\end{definition}

\Cref{fig:G3pc} defines the well-known sequent calculus~$\LK$ for classical logic~$\CPC$ and the single-conclusion calculus~$\LJ$ for intuitionistic logic~$\IPC$. We define the calculi without the cut rule~(\cref{fig:G3pc}). \Cref{fig:modal-rules} presents some well-known modal rules, and \cref{fig:modal-sequent} defines calculi for classical modal logics. 
We denote sequent calculi in bold font, using the name of the logic, sometimes prefixed with~$\SC\mathbf{.}$, to identify the system. Thus,
in \cref{fig:modal-rules} we mean that $\GthreeK$~is a sequent calculus for modal logic~$\K$,~etc.

\begin{figure}[t]
\centering{
\small
\fbox{
    \parbox{0.965\textwidth}{
    \vspace{-1.2em}
    \[
        \vlinf{\id}{}{p \seqar p}{}
    \]
    \vspace{-1.9em}
    \\
    \[
        \vlinf{\lr\bot}{}{\bot\seqar}{}
        \qquad
        \vliinf{\lr\impl}{}{\phi \impl \psi,\Gamma\seqar\Delta}{\Gamma\seqar\Delta,\phi}{\psi, \Gamma \seqar \Delta} 
        \qquad
        \vliinf{\lr\vee}{}{\phi\vee\psi,\Gamma\seqar\Delta}{\phi,\Gamma\seqar\Delta}{\psi,\Gamma\seqar\Delta} 
        \qquad
        \vlinf{\lr\wedge^i }{\mathsmaller{(i=0,1)}}{\phi_0\wedge\phi_1,\Gamma\seqar\Delta}{\phi_i,\Gamma\seqar\Delta}
    \]
    \vspace{-0.8em}
    \\
    \[
        \vlinf{\rr\top}{}{\seqar\top}{}
        \qquad
        \vlinf{\rr\impl}{}{\Gamma\seqar\Delta,\phi \impl \psi}{\phi,\Gamma\seqar\Delta, \psi} 
        \qquad
        \vlinf{\rr\vee^i}{\mathsmaller{(i=0,1)}}{\Gamma\seqar\Delta,\phi_0 \vee \phi_1}{\Gamma\seqar\Delta,\phi_i} 
        \qquad
        \vliinf{\rr\wedge}{}{\Gamma\seqar\Delta,\phi\wedge\psi}{\Gamma\seqar\Delta,\phi}{\Gamma\seqar\Delta,\psi}
    \]
    \vspace{-0.8em}
    \\
    \[
        \vlinf{\lr\wk}{}{\phi,\Gamma\seqar\Delta}{\Gamma\seqar\Delta}
        \qquad \qquad
        \vlinf{\rr\wk}{}{\Gamma\seqar\Delta,\phi}{\Gamma\seqar\Delta}
        \qquad \qquad 
        \vlinf{\lr\cntr}{}{\phi,\Gamma\seqar\Delta}{\phi,\phi,\Gamma\seqar\Delta}
        \qquad \qquad
        \vlinf{\rr\cntr}{}{\Gamma\seqar\Delta,\phi}{\Gamma\seqar\Delta,\phi,\phi}
    \]
    \vspace{0.1em}
    \hrule
    \vspace{-0.2em}
    \[
        \qquad \qquad \qquad \qquad \qquad \qquad \vliinf{\cut}{}{\Gamma,\Gamma' \seqar \Delta,\Delta'}{\Gamma \seqar \Delta, \phi}{\phi,\Gamma'\seqar\Delta'}
    \]
    \vspace{-1em}
    }
}}
\caption{\textbf{Top:} Sequent calculus~$\LK$. Subscripts~$l$~and~$r$ stand for \emph{left} and \emph{right} rule respectively. The \emph{structural rules} are \emph{weakening}~$(\lr\wk)$ and $(\rr\wk)$ and \emph{contraction}~$(\lr\cntr)$ and~$(\rr\cntr)$. Sequent calculus~$\LJ$ is obtained by restricting these rules to single-conclusion sequents, i.e.,~$\Delta=\emptyset$ in the right rules, and in the left rules $\Delta$~contains at most one formula, with the extra condition for~$(\lr\impl)$ that $\Delta$~is not present in the left premise, and the right contraction rule~$(\rr\cntr)$ is omitted. \textbf{Bottom:} The cut rule. For single-conclusion calculi we require that $\Delta = \emptyset$ and that $\Delta'$~contains at most one formula.}
\label{fig:G3pc}
\end{figure}

\begin{figure}[t!]
  \centering
  \begin{minipage}[t]{0.6\textwidth}
    \small{
\fbox{\parbox{.97\textwidth}{
{\centering
\vspace{-0.22em}
\[
\vlinf{\Kseqrule}{}{\Box\Gamma\seqar\Box \phi}{\Gamma \seqar \phi} \quad \quad 
\vlinf{\Tseqrule}{}{\Box \phi,\Gamma\seqar\Delta}{\Box \phi,\phi,\Gamma \seqar \Delta} \qquad \quad
\vlinf{\Dseqrule}{}{\Box\Gamma, \Box \phi \seqar}{\Gamma,\phi\seqar}
\]
\\
\[
\vlinf{\Fourseqrule}{}{\Box\Gamma\seqar\Box \phi}{\Box \Gamma, \Gamma \seqar \phi} \quad \quad 
\vlinf{\Sfourseqrule}{}{\Box\Gamma\seqar \Box \phi}{\Box \Gamma \seqar \phi} \quad \quad
\vlinf{\GLseqrule}{}{\Box\Gamma\seqar\Box \phi}{\Box \Gamma, \Gamma,\Box \phi \seqar \phi} 
\]
\vspace{-0.25em}
}
}}}
\caption{Modal sequent rules.} \label{fig:modal-rules}
  \end{minipage}
  \hfill
  \begin{minipage}[t]{0.37\textwidth}
    \small{
    \fbox{\parbox{0.95\textwidth}{
    \vspace{-1.15em}
    \begin{align*}
        \GthreeK &\ce  \LK + (\Kseqrule)\\
        \GthreeT &\ce  \GthreeK + (\Tseqrule)\\
        \GthreeD &\ce  \GthreeK + (\Dseqrule)\\
        \GthreeKfour &\ce  \LK + (\Fourseqrule)\\
        \GthreeSfour &\ce  \LK + (\Sfourseqrule) + (\Tseqrule)\\
        \GthreeGL &\ce  \LK + (\GLseqrule)
    \end{align*}
    \vspace{-1.75em}
    }}
    }
\caption{Modal sequent calculi.} \label{fig:modal-sequent}
  \end{minipage}
\end{figure}

Important to note is that the cut rule 
is admissible in the presented calculi. It means that the rule can safely be added to the calculi without deriving new sequents.
    Formally, a rule with premises $S_1, \dots, S_n$ and conclusion~$S$ is \emph{admissible} in a sequent calculus~$\SC$ if{f} for any substitution~$\sigma$, if $\SCproves \sigma{S_i}$ for all~$i$, then $\SCproves \sigma{S}$.  

\begin{theorem}[Cut-admissibility]
    Rule $(\cut)$ is admissible in all calculi from \cref{fig:G3pc,fig:modal-sequent}.
\end{theorem}

    Cut-admissibility theorems feature prominently in proof theory. Indeed, the cut rule allows for easy completeness results, but the possibility to eliminate cut has far reaching applications such as consistency and decidability results. Cut-free calculi are used in Maehara's proof-theoretic method for the~CIP~\cite{Maehara61}. We will see in \cref{sect:role_of_cut} that some forms of cut can be adopted in Maehara's method, but may create problems for the LIP.

\begin{remark}[Design choices of sequent calculi]
\label{remark:design_choices}
Multiple choices can be made when designing a sequent calculus. Two common choices  are: the type of sequents used (e.g.,~lists, multisets, or sets), and the choice between explicit or implicit structural rules. Generally, these choices make no difference for Maehara's method to prove the~CIP, only the cut rule may present a problem.  
Pitts' method for the~UIP  is more sensitive to such design choices where contraction rules create a problem~(\cref{subsec:Pitts}).  \lipicsEnd
\end{remark}

\begin{remark}[Sequent calculi vs tableaux]\label{remark:tableaux}
    The intended reading of the sequent rules is from top to bottom: assuming the validity of the premises, the conclusion is valid. Dual to sequent calculi are tableau calculi~\cite{Fitting83}. These can be viewed as sequent calculi `upside down', where the unsatisfiability of the upper formulas is preserved for the lower formulas. Maehara style methods also exist for tableau calculi~\cite{Fitting83}. See  \refchapter{chapter:propositional} of this book for more details. \lipicsEnd
\end{remark}

\section{Interpolation Using Sequent Calculi}
\label{sec:basics}

In this section we discuss the basics of the proof-theoretic methods of interpolation properties based on propositional and modal logics.

\subsection{Craig and Lyndon Interpolation via Maehara's Method}
\label{subsec:Maehara}

There are different types of Craig interpolation depending on what consequence is under consideration (see other chapters in this book). Maehara's method is very suitable for the \emph{local} Craig interpolation property where we interpolate the local consequence relation represented by $\impl$. It is also suitable for the local Lyndon interpolation property, which takes into account the polarity of propositional variables.

\begin{definition}[Signature]
The \emph{signature} of a formula~$\phi$, denoted~$\Voc{\phi}$, is the set of propositional variables occurring in~$\phi$. 

For the propositional language~$\calLp$ and modal language~$\calLBox$ we inductively define the sets~$\Vocpos{\phi}$~and~$\Vocneg{\phi}$ of propositional variables occurring in~$\phi$ with \emph{positive} and \emph{negative polarity} respectively: $\Vocpos{p}\ce\{p\}$ and $\Vocneg{p}\ce\emptyset$ for $p \in \Prop$; $\Vocpos{\bot} = \Vocpos{\top} = \Vocneg{\top} = \Vocneg{\bot}\ce\emptyset$; connectives $\vee$ and $\wedge$ preserve polarity; connective $\impl$ flips polarity in the antecedent and preserves polarity in the succedent; modality~$\Box$ preserves polarity. 

For a multiset~$\Gamma$, we define $\Voc{\Gamma}\ce\bigcup_{\phi\in\Gamma}\Voc{\phi}$, and analogously  for $\Vocpos{\Gamma}$ and $\Vocneg{\Gamma}$. For a sequent $\Gamma \seqar \Delta$, we define $\Vocpos{\Gamma \seqar \Delta} \ce \Vocneg{\Gamma} \cup \Vocpos{\Delta}$ and $\Vocneg{\Gamma \seqar \Delta} \ce \Vocpos{\Gamma} \cup \Vocneg{\Delta}$, where polarity flips in the antecedent similarly to~$\impl$.
\end{definition}

It is clear from the definition that $\Voc{\phi}=\Vocpos{\phi} \cup \Vocneg{\phi}$ for any formula $\phi$ from~$\calLp$~or~$\calLBox$, and similarly for multisets and sequents.

\begin{example}
$\Vocpos{(\Box p \impl q) \impl r} = \{p,r\}$ and $\Vocneg{(\Box p \impl q) \impl r} = \{ q \}$. 
 Similarly, $\Vocpos{(\Box p \impl q) \impl q} = \{p,q\}$ and $\Vocneg{(\Box p \impl q) \impl q} = \{ q \}$.  \lipicsEnd
\end{example}

\begin{remark}
We define the signature to consist of only propositional variables. One can also consider modalities as part of the signature. In the setting of first-order logic, the signature consists of relational and function symbols occurring in a first-order formula~$\phi$. \lipicsEnd
\end{remark}

\begin{definition}[Craig interpolation]
\label{def:Craig-interpolation}
A logic~$\lgc$ over a language~$\calL$ has the 
\emph{Craig Interpolation Property~(CIP)}, if{f} whenever $\Lproves \phi \impl \psi$, there is a formula $\theta \in \form_{\calL}$ such that:
    \begin{enumerate}
        \item\label{cond:seq_var} $\Voc{\theta}\subseteq \Voc{\phi} \cap \Voc{\psi}$,
        \item $\Lproves\phi \impl \theta$ and $\Lproves\theta\impl\psi$.  
    \end{enumerate}
Such formula~$\theta$ is called a  \emph{(Craig) interpolant} of $\phi \impl \psi$ in~$\lgc$.
\end{definition}

\begin{definition}[Lyndon interpolation]
    The \emph{Lyndon Interpolation Property~(LIP)} for logics over language~$\calLp$~or~$\calLBox$ is defined analogously to the~CIP  replacing condition~\ref{cond:seq_var} in \cref{def:Craig-interpolation} with the following condition:
    \begin{enumerate}[1'.]
    \item 
    \label{item:Lyndon}
    $\Vocpos{\theta}\subseteq \Vocpos{\phi} \cap \Vocpos{\psi}$ and $\Vocneg{\theta}\subseteq \Vocneg{\phi} \cap \Vocneg{\psi}$.
\end{enumerate}
In this case~$\theta$ is called a  \emph{Lyndon interpolant} of $\phi \impl \psi$ in~$\lgc$.
\end{definition}

It is clear from the definition that the~LIP  implies the~CIP.

\begin{example}[Non-unique interpolants]
\label{ex:non-unique}
Given a provable implication $\phi \impl \psi$, there may exist multiple provably non-equivalent Craig/Lyndon interpolants. For example in~$\CPC$, let $\phi =  (p \wedge q) \vee (\neg r \wedge s)$ and $\psi =  t \vee p \vee q \vee \neg r$. We have $\proves_\CPC \phi \impl \psi$. There are multiple provably non-equivalent interpolants for it, namely $(p \wedge q) \vee \neg r$, $p \vee \neg r$, $q \vee \neg r$, and $p \vee q \vee \neg r$. \lipicsEnd
\end{example}

Now we introduce Maehara's method and show how sequent calculi can be used to prove the~CIP~and~LIP. \emph{Maehara's method} is commonly described as the proof-theoretic method of proving the~CIP by an induction on a sequent proof (or proof in another suitable proof system). This is a helpful intuition. However, the method is more refined and interesting because of its more complex underlying process. The idea is to formulate a sequent-based interpolation property that is at least as strong as the~CIP and that relies on a calculus with richer sequents called split sequents.

\begin{definition}[Split sequent]
\label{def:split_seq}
    The expression of the form $\Gamma ; \Gamma' \seqar \Delta ; \Delta'$ is called a \emph{split sequent}. Multisets~$\Gamma$~and~$\Delta$ belong to the \emph{left split}, and $\Gamma'$ and $\Delta'$ to the \emph{right split}.
\end{definition}

One can also think of sequent $\Gamma \seqar \Delta$ being the left split and sequent $\Gamma' \seqar \Delta'$ being the right split of split sequent $\Gamma ; \Gamma' \seqar \Delta ; \Delta'$.

\begin{definition}[Sequent Craig interpolation]\label{def:Maehara}
Let $\SC$~be a sequent calculus over a language~$\calL$. We say that $\SC$~has the \emph{Sequent Craig Interpolation Property~(SCIP)} if{f} whenever we have $\SCproves \Gamma, \Gamma' \seqar \Delta, \Delta'$, there exists a formula $\theta \in \form_\calL$ such that
\begin{enumerate}
    \item\label{cond:seq_lang} $\Voc{\theta}\subseteq\Voc{\Gamma,\Delta}\cap\Voc{\Gamma',\Delta'}$, and
    \item\label{cond:seq_main} 
     $\SCproves \Gamma \seqar \Delta, \theta$ and $\SCproves \theta, \Gamma' \seqar \Delta'$.
\end{enumerate}  
If above properties hold we write $\SCproves \Gamma ; \Gamma' \xseqar{\theta}\Delta ; \Delta'$ and we call~$\theta$ a \emph{(Craig) interpolant} for split sequent $\Gamma ; \Gamma' \seqar \Delta ; \Delta'$ in~$\SC$. When dealing with a single-conclusion calculus we require $\Delta = \emptyset$ and $\Delta'$~to contain at most one formula.
 \end{definition}

The intuition of the~SCIP is that formula~$\theta$ interpolates between the left and right side of the split, rather than between the antecedent and succedent of the sequent. The splits are essential as explained after \cref{thm:Craig-sequent->logic}.

For sequent Lyndon interpolation we take into account the polarity of the variables:

\begin{definition}[Sequent Lyndon interpolation]
\label{def:sequent_Lyndon}
The \emph{Sequent Lyndon Interpolation Property~(SLIP)} for logics over language~$\calLp$~or~$\calLBox$ is defined analogously to the~SCIP replacing condition~\ref{cond:seq_lang} in \cref{def:Maehara}~with:
\begin{enumerate}[1'.]
    \item\label{cond:lynd_one} $\Vocpos{\theta}\subseteq \Vocneg{\Gamma \seqar \Delta}\cap \Vocpos{\Gamma' \seqar \Delta'}$  and $\Vocneg{\theta}\subseteq \Vocpos{\Gamma \seqar \Delta}\cap \Vocneg{\Gamma' \seqar \Delta'}$.
\end{enumerate}
In this case~$\theta$ is called a  \emph{Lyndon interpolant} for split sequent $\Gamma ; \Gamma' \seqar \Delta ; \Delta'$ in~$\SC$.
\end{definition}

The interpolation property of a sequent calculus implies interpolation of its logic:

\begin{theorem}\label{thm:Craig-sequent->logic} 
If a sequent calculus~$\SC$ has the~SCIP/SLIP, then its logic~$\lgc$ has the~CIP/LIP.
\end{theorem}
\begin{proof}[Proof sketch]
Assume that $\SC$~has the~SCIP/SLIP and let $\Lproves\phi \impl\psi$. By \cref{CalculusForLogic}, $\SCproves \phi \seqar \psi$. For any Craig/Lyndon interpolant $\theta$  for the split $\phi ; \ \seqar \ ; \psi$ of this sequent in~$\SC$, one can prove that $\theta$~is a Craig/Lyndon interpolant of $\phi \impl \psi$ in~$\lgc$.
\end{proof}

The idea of Maehara's method is to present the sequent rules using split sequents, where each sequent rule has multiple counterparts in the split sequent calculus depending on whether the principal formula is on the left or on  the right side of the split. \Cref{fig:interpolation} presents the split system of $\LK$. Given a split sequent proof of $\SCproves \phi ; \ \seqar \ ; \psi$, one assigns Craig interpolants to the leaves and inductively defines the interpolants when going down the tree yielding an interpolant for $\SCproves \phi ; \ \seqar \ ; \psi$. The splits ensure the variable property~\ref{cond:seq_lang} of \cref{def:Maehara}. Indeed, bottom-up, the splits keep track of those formulas originating from $\phi$ on the left and those originating from $\psi$ on the right side of the split.  For example, in the following instance of rule~$R(\rr\impl)$, formula $\psi_1$ should be stored on the right side of the split in the antecedent of the premise:\looseness=-1
\[
{\small{
\vlinf{R(\rr\impl)}{}{\phi ; \seqar\ ;\psi_1 \impl \psi_2}{ \phi ; \psi_1 \seqar \ ; \psi_2}
}}
\]
In the split sequent calculus of $\LK$ we see that active formulas are always on the same side as the principal formula, which is not necessarily the case for sequent calculi for other logics, see \cref{remark:IPC}.

\begin{remark}[Monotone connectives]
When working with a language with only \emph{monotone connectives}, Maehara's split is redundant as the left and right sides can be taken as the left and right sides of the sequent arrow~$\seqar$  (cf.~symmetric sequents in~\cite{Fitting83}).  
\lipicsEnd
\end{remark}

\begin{figure}[!t]
\centering{
\small{
\fbox{\parbox{0.97\textwidth}{
\vspace{0.5em}
\ \ \textbf{Split axioms:}
\vspace{-0.5em}
\[
        \vlinf{\itRR(\id)}{}{;p\xseqar{\top} \ ;p}{} \quad
        \vlinf{\itLR(\id)}{}{p;\ \xseqar{p}\ ;p}{} \quad
        \vlinf{\itRL(\id)}{}{;p\xseqar{\neg p}p;}{} \quad
        \vlinf{\itLL(\id)}{}{p; \ \xseqar{\bot}p;}{}
\]
\vspace{-0.9em}
\[
        \vlinf{L(\lr\bot)}{}{\bot; \ \xseqar{\bot} \ ;}{} \quad \ \ 
        \vlinf{R(\lr\bot)}{}{;\bot\xseqar{\top} \ ;}{} \quad \ 
        \vlinf{L(\lr\top)}{}{; \ \xseqar{\bot}\top;}{} \quad \ \ 
        \vlinf{R(\rr\top)}{}{; \ \xseqar{\top} \ ; \top}{} 
\]

\ \ \textbf{Some unary split rules:}
\[
\vlinf{R(\lr\wedge^i)}{}{\Gamma ; \Gamma',\phi_0\wedge\phi_1\xseqar{\theta}\Delta; \Delta'}{\Gamma; \Gamma',\phi_i\xseqar{\theta}\Delta; \Delta'} \quad
\vlinf{R(\rr\impl)}{}{\Gamma;\Gamma'\xseqar{\theta}\Delta;\Delta',\phi \impl \psi}{\Gamma;\Gamma',\phi\xseqar{\theta}\Delta;\Delta',\psi} \quad\vlinf{L(\rr\cntr)}{}{\Gamma;\Gamma'\xseqar{\theta}\Delta,\phi;\Delta'}{\Gamma;\Gamma' \xseqar{\theta} \Delta,\phi,\phi;\Delta'}
\]

\ \ \textbf{Binary split rules:}
\[
        \vliinf{L(\lr\vee)}{}{\phi\vee\psi,\Gamma; \Gamma'\xseqar{\theta_1 \vee \theta_2}\Delta;\Delta'}{\phi,\Gamma; \Gamma'\xseqar{\theta_1}\Delta;\Delta'}{\psi,\Gamma; \Gamma'\xseqar{\theta_2}\Delta;\Delta'} \ \ \  
        \vliinf{R(\lr\vee)}{}{\Gamma; \Gamma',\phi\vee\psi\xseqar{\theta_1 \wedge \theta_2}\Delta;\Delta'}{\Gamma; \Gamma',\phi\xseqar{\theta_1}\Delta;\Delta'}{\Gamma; \Gamma',\psi\xseqar{\theta_2}\Delta;\Delta'} 
\]
\[
\vliinf{L(\rr\wedge)}{}{\Gamma;\Gamma'\xseqar{\theta_1 \vee \theta_2}\phi\wedge\psi,\Delta;\Delta'}{\Gamma;\Gamma'\xseqar{\theta_1}\phi,\Delta;\Delta'}{\Gamma;\Gamma'\xseqar{\theta_2}\psi,\Delta;\Delta'} \quad
\vliinf{R(\rr\wedge)}{}{\Gamma;\Gamma'\xseqar{\theta_1 \wedge \theta_2}\Delta;\Delta',\phi\wedge\psi}{\Gamma;\Gamma'\xseqar{\theta_1}\Delta;\Delta',\phi}{\Gamma;\Gamma'\xseqar{\theta_2}\Delta;\Delta',\psi}
\]
\[
\vliinf{L(\lr\impl)}{}{\phi \impl \psi,\Gamma;\Gamma'\xseqar{\theta_1 \vee \theta_2}\Delta;\Delta'}{\Gamma;\Gamma'\xseqar{\theta_1}\Delta,\phi;\Delta'}{\psi, \Gamma;\Gamma' \xseqar{\theta_2} \Delta;\Delta'} \quad
\vliinf{R(\lr\impl)}{}{\Gamma;\Gamma',\phi \impl \psi\xseqar{\theta_1 \wedge \theta_2}\Delta;\Delta'}{\Gamma;\Gamma'\xseqar{\theta_1}\Delta;\Delta',\phi}{\Gamma;\Gamma',\psi \xseqar{\theta_2} \Delta;\Delta'} \quad
\]
}
}}
}
\caption{Craig interpolant construction with split sequent calculus for $\LK$. Capital letters~$L$~and~$R$ denote whether the principal formula is on the left or right side of the split respectively.  The figure presents only three of the unary split rules, because all other unary split rules have the same shape in which the active formulas in the premise stay on the same side of the split as the principal formula of the conclusion and the interpolant of the conclusion is preserved from the premise.}
\label{fig:interpolation}
\end{figure}

We present Maehara's method for classical propositional logic: 

\begin{theorem}\label{thm:Maehara_CPC}
    $\LK$ has the~SCIP~and~SLIP, and hence $\CPC$ enjoys the~CIP~and~LIP.
\end{theorem}
\begin{proof}
Suppose $\proves_\LK \Gamma; \Gamma' \seqar \Delta; \Delta'$ and let $\pi$ be its split sequent proof with rules from \cref{fig:interpolation}. We will construct an interpolant $\theta$ by induction on the structure of $\pi$. For both the base step and the induction step, there are several subcases based on the lowest split sequent rule of $\pi$. The interpolant construction for each case is summarized in \cref{fig:interpolation}. Here we will check two cases. As most textbooks only consider Craig interpolants, here we show that this construction also satisfies the stronger variable condition for Lyndon interpolants.

\begin{itemize}
\item $\itRR(\id$): Suppose 
\vspace{-0.5em}
\[
{\small{
\pi \ = \ \vlinf{\itRR(\id)}{}{;p\seqar \ ;p}{}
}}
\]
We claim that $\theta=\top$. Indeed, the variable condition~\hyperref[cond:lynd_one]{$1'$} from \cref{def:sequent_Lyndon} holds because $\Voc\top=\emptyset$. Also, $\proves_\LK \ \seqar \top$ by~$(\rr\top)$ and $\proves_\LK \top, p \seqar p$ by~$(\id)$~and~$(\lr\wk)$ as desired.

   \item $L(\lr\impl)$: Suppose that the bottommost rule is an instance of~$(\lr\impl)$ with the principal formula on the left side of the split. For the premises the appropriate splits are chosen.
\vspace{-0.5em}
        \[
{\small{
       \pi \ = \ \vlderivation{\vliin{L(\lr\impl)}{}{\phi\impl\psi,\Gamma; \Gamma'\seqar\Delta;\Delta'}{\stub{\pi_1}{\Gamma; \Gamma'\seqar\Delta,\phi;\Delta'}}{\stub{\pi_2}{\psi,\Gamma; \Gamma'\seqar\Delta;\Delta'}}}
}}
\]
By the induction hypothesis (IH) on $\pi_1$ and $\pi_2$, there are Lyndon interpolants~$\theta_1$~and~$\theta_2$ such that $\Gamma;\Gamma'\xseqar{\theta_1}\Delta,\phi;\Delta'$ and $\psi,\Gamma;\Gamma'\xseqar{\theta_2}\Delta;\Delta'$. We prove that $\theta=\theta_1 \vee \theta_2$ is the desired interpolant. Out of the two variable conditions from \cref{def:sequent_Lyndon}, we show how to check the first (the other is symmetric):
	\begin{align*} 
        &\Vocpos{\theta_1 \vee \theta_2}=\Vocpos{\theta_1}\cup\Vocpos{\theta_2}\\
        &\quad \subseteq \Big(\Vocneg{\Gamma\seqar\Delta,\phi}\cap\Vocpos{\Gamma'\seqar\Delta'}\Big) \cup \Big(\Vocneg{\psi,\Gamma \seqar \Delta} \cap \Vocpos{\Gamma'\seqar\Delta'}\Big)\, (\text{IH})\\
        &\quad = \Big(\Vocneg{\Gamma\seqar\Delta,\phi}\cup\Vocneg{\psi,\Gamma \seqar \Delta}\Big) \cap\Vocpos{\Gamma'\seqar\Delta'}\\
        &\quad =\Big(\Vocpos{\psi,\Gamma}\cup\Vocneg{\Delta,\phi}\Big)\cap\Vocpos{\Gamma'\seqar\Delta'}\\
        &\quad =\Big(\Vocpos{\phi\impl\psi,\Gamma} \cup \Vocneg{\Delta}\Big)\cap\Vocpos{\Gamma'\seqar\Delta'}\\
        &\quad = \Vocneg{\phi\impl\psi,\Gamma \seqar \Delta} \cap\Vocpos{\Gamma'\seqar\Delta'}.
    \end{align*}
	
It remains to check that $\proves_\LK \phi \impl \psi, \Gamma \seqar \Delta, \theta_1 \vee \theta_2$ and $\proves_\LK \theta_1 \vee \theta_2, \Gamma' \seqar \Delta'$. Indeed:
\[
{\small{
\vlderivation{
\vlin{\rr\cntr,\rr\vee^0,\rr\vee^1}{}{\phi\impl\psi,\Gamma\seqar\Delta,\theta_1 \vee \theta_2}{
    \vliin{\lr\impl}{}{\phi\impl\psi,\Gamma\seqar\Delta,\theta_1,\theta_2}{
    		\vlin{\rr\wk}{}{\Gamma\seqar\Delta,\theta_1,\theta_2,\phi}{
        		\stub{\mathit{IH}}{\Gamma\seqar\Delta,\theta_1,\phi}
			}
    }
    {
    		\vlin{\rr\wk}{}{\psi,\Gamma\seqar\Delta,\theta_1,\theta_2}{
        		\stub{\mathit{IH}}{\psi,\Gamma\seqar\Delta,\theta_2}
			}
    }
}
}
\qquad \quad
\vlderivation{
\vliin{\lr\vee}{}{\theta_1 \vee \theta_2,\Gamma'\seqar\Delta'}{
    \stub{\mathit{IH}}{\theta_1,\Gamma'\seqar\Delta'}
}{
    \stub{\mathit{IH}}{\theta_2,\Gamma'\seqar\Delta'}
}
}
}}
\]
\end{itemize}
We leave the other cases to the reader.
\end{proof}

\begin{remark}[Computation and complexity]
    Maehara's method is \emph{constructive}: it provides an algorithm to compute interpolants. The complexity of an interpolant of $\phi \impl \psi$ is \emph{linear} in the number of sequents occurring in the proof-tree of $\phi \seqar \psi$. Different proofs of a sequent may yield different interpolants. \lipicsEnd
\end{remark}

\begin{remark}[Incompleteness]
    Interestingly, it has been shown that Maehara's method is \emph{incomplete} in the sense that there are implications $\phi \impl \psi$ for which there are interpolants that cannot be computed by Maehara's method \cite{HetzlJalali24} (see also the next example). \lipicsEnd
\end{remark}

\begin{example}
    Recall \cref{ex:non-unique} with $\phi = (p \wedge q) \vee (\neg r \wedge s)$ and $\psi = t \vee p \vee q \vee \neg r$. We work in the split calculus $\LK$ (\cref{fig:interpolation}). Craig interpolant $p \vee \neg r$ is derived as follows:
    \vspace{-1em}
    \[
    {\small{
\vlderivation{
\vliin{L(\lr\vee),R(\rr\vee)}{}{(p\wedge q) \vee (\neg r \wedge s) ; \ \xseqar{p \vee \neg r} \ ; t \vee p \vee q \vee \neg r}{
    \vlin{L(\lr\wedge)}{}{p \wedge q ; \ \xseqar{p} \ ; t, p, q, \neg r}{
        \vlin{L(\lr\wk),R(\rr\wk)}{}{p, q ; \ \xseqar{p} \ ; t, p, q, \neg r}{
            \vlin{\itLR(\id)}{}{p ; \ \xseqar{p}\ ; p}{\vlhy{}}
        }
    }
}
{
    \vlin{L(\lr\wedge)}{}{\neg r \wedge s ; \ \xseqar{\neg r} \ ;  t,p, q, \neg r}{
        \vlin{L(\lr\wk),R(\rr\wk)}{}{\neg r, s ; \ \xseqar{\neg r} \ ; t, p, q, \neg r}{
            \vliin{L(\lr\impl)}{}{\neg r ; \ \xseqar{\neg r \vee \bot (= \neg r)} \ ; \neg r}{
                \vlin{R(\rr\impl),R(\rr\wk)}{}{\ \xseqar{\neg r} r ; \neg r}{
                    \vlin{\itRL(\id)}{}{ ; r \xseqar{\neg r} r ; }{\vlhy{}}
                }
            }
            {
                \vlin{R(\rr\wk)}{}{\bot ; \ \xseqar{\bot} \ ; \neg r}{
                    \vlin{L(\lr\bot)}{}{\bot ; \ \xseqar{\bot} \ ; }{\vlhy{}}
                }
            }   
        }
    }
}
}
}}
\]
A similar split sequent proof with the left branch starting with axiom $q ; \xseqar{q} ; q$ yields interpolant $q \vee \neg r$. Interpolant $(p \vee q) \vee \neg r$ can  be obtained by Maehara's method using contraction rules. However, interpolant $(p \wedge q) \vee \neg r$ cannot be obtained by Maehara's method in~$\LK$ (see~\cite{HetzlJalali24}). \looseness=-1\lipicsEnd
\end{example}

\begin{remark}[Other types of sequent calculi]
Maehara's method is very flexible and has been used for different types of sequent calculi (recall \cref{remark:design_choices}). We worked with two-sided (multiset) sequent calculus~$\LK$.  It equally works for sequents based on lists or sets, and when working with implicit weakening and contraction rules~\cite{TroelstraS00}. The method also works for \emph{one-sided sequent} calculi based on formulas in negation normal form, where a one-sided split sequent has the form $\Delta_l ; \Delta_r$ with left split~$\Delta_l$ and right split~$\Delta_r$.  
Recall from \cref{remark:tableaux} that sequent calculi can be viewed as the dual of tableau calculi. For that reason, Maehara's method can be adjusted to tableaux (see  \refchapter{chapter:propositional}.10 of this book). \lipicsEnd
\end{remark}

\begin{remark}[Single-conclusion calculi]
\label{remark:IPC}
Maehara's method also applies to single-conclusion sequent calculi such as~$\LJ$. The main difference is that sequents have only one formula in the succedent, so there is no need to split the succedent; equivalently, one may view the succedent as always belonging to the ``right'' side of the split. For~$\LJ$, the construction of interpolants proceeds much as in the $\LK$~case described in \cref{fig:interpolation}, except that rules such as $\itRL(\id)$ and $R(\rr\wedge)$ are omitted, and the split rule $L(\lr\impl)$ is adapted as follows:
\[
{\small{
\vliinf{L(\lr\impl)}{}{\phi \impl \psi,\Gamma;\Gamma' \xseqar{\theta_1 \impl \theta_2} \Delta'}{\Gamma';\Gamma \xseqar{\theta_1} \phi}{\psi, \Gamma;\Gamma' \xseqar{\theta_2} \Delta'}
}}
\]
where the formula $\phi \impl \psi$ is on the left side of the split in the conclusion, but $\phi$~is forced to be on the `right' side of the split in the left premise. This forces us to swap the sides of~$\Gamma'$~and~$\Gamma$ in this premise, to make sure that $\Gamma$~and~$\phi$~remain on the same side of the split. \lipicsEnd
\end{remark}

In light of the previous remark we have the following theorem.

\begin{theorem}
    Logic~$\IPC$ has the~CIP~and~LIP.
\end{theorem}

\begin{figure}
\small{
\fbox{\parbox{.97\textwidth}{
{\centering
\[
\vlinf{L(\Kseqrule)}{}{\Box\Gamma;\Box\Gamma'\xseqar{\Diamond\theta}\Box \phi;}{\Gamma;\Gamma' \xseqar{\theta} \phi;} \,\,\,\,\,
\vlinf{R(\Kseqrule)}{}{\Box\Gamma;\Box\Gamma'\xseqar{\Box\theta};\Box \phi}{\Gamma;\Gamma' \xseqar{\theta} ;\phi} \,\,\,\,\,\,
\vlinf{L(\Dseqrule)}{}{\Box\phi,\Box\Gamma;\Box \Gamma' \xseqar{\Box \theta}\, ;}{\phi,\Gamma;\Gamma'\xseqar{\theta}\,;}
\,\,\,\,\,
\vlinf{R(\Dseqrule)}{}{\Box\Gamma;\Box \Gamma',\Box \phi \xseqar{\Box\theta}\, ; }{\Gamma;\Gamma',\phi\xseqar{\theta}\,;}
\]
\\ \vspace{-1em}
\[
\vlinf{L(\GLseqrule)}{}{\Box\Gamma;\Box\Gamma'\xseqar{\Diamond\theta}\Box \phi;}{\Box \Gamma, \Gamma,\Box \phi ; \Box \Gamma', \Gamma'\xseqar{\theta} \phi;} \quad
\vlinf{R(\GLseqrule)}{}{\Box\Gamma;\Box\Gamma'\xseqar{\Box\theta};\Box \phi}{\Box \Gamma, \Gamma; \Box \Gamma', \Gamma',\Box \phi \xseqar{\theta} ;\phi}
\quad
\vlinf{L(\Tseqrule)}{}{\Box \phi,\Gamma;\Gamma'\xseqar{\theta}\Delta;\Delta'}{\phi,\Box \phi,\Gamma;\Gamma' \xseqar{\theta} \Delta;\Delta'} 
\]
}
\vspace{-1em}
}}}
\caption{Interpolant transformations for modal rules. When the principal formula is on the left side of the split in rules~$(\Kseqrule)$~and~$(\GLseqrule)$, it is a `$\Diamond$-transformation', and when it is on the right it is a `$\Box$-transformation'. It similarly works for rules~$(\Fourseqrule)$~and~$(\Sfourseqrule)$ from \cref{fig:modal-rules}.  For rule~$(\Tseqrule)$,  the same interpolant suffices for both splits. Both transformations for~$(\Dseqrule)$ are `$\Box$-transformation' because the only role $\phi$ plays there is to ensure that the antecedent is not empty. See  \cref{ex:unlabelled-modal-rules,Ex: box diamond}.\looseness=-1} 
\label{fig:modal_Craig_transformations}
\end{figure}

For modal logics, we provide the interpolant transformations for the sequent calculi in \cref{fig:modal-sequent}, and show that logics~$\K$, $\T$, $\D$, $\Kvier$, $\Svier$,~and~$\GL$ have the~CIP. In fact, all of these except for~$\GL$ admit the~LIP. When considering  split rules~$L(\GLseqrule)$~and~$R(\GLseqrule)$ from \cref{fig:modal_Craig_transformations}, $\phi$~and both~$\Box \phi$'s~are placed on the same side of the split. Thus, the Craig condition in \cref{def:Maehara}\eqref{cond:seq_lang} is satisfied, but the stronger Lyndon one in \cref{def:sequent_Lyndon}($1'$) cannot be guaranteed: $\Box \phi$~in the premise antecedent yields the opposite polarities from~$\Box \phi$ in the conclusion succedent. Consequently, sequent calculus~$\GthreeGL$ can be used to prove the~CIP but not the~LIP, despite the fact that logic~$\GL$ does have the~LIP~\cite{Shamkanov11}. An existing proof-theoretic proof of the~LIP for~$\GL$~\cite{Shamkanov14} uses cyclic proofs, which are discussed in  \refchapter{chapter:fixedpoint}. 
More generally, the existence of a cut-free sequent calculus is no guarantee that the logic has the~CIP, nor the logic having the~CIP is a guarantee that a cut-free sequent calculus can be used to prove it. \looseness=-1

Although Maehara's method is flexible across many types of logics, it does not mean that  many logics of each type enjoy the~CIP. In fact, interpolation is rare among, e.g.,~intermediate and modal logics~\cite{Ghil2,GabbayM05}, which has crucial impact on proof theory (see~\cref{sect:results}). \looseness=-1

\subsection{Restricted Cuts and Interpolation}
\label{sect:role_of_cut}

In many textbooks on proof theory~\cite{TroelstraS00,Takeuti75}, Maehara's method is presented as an application of the cut-admissibility theorem. In the previous section we assumed to have a cut-free calculus, but it is well-known that special forms of cut can be admitted for Maehara's method to prove the~CIP (though not necessarily the~LIP). This is especially useful for logics that do not admit a cut-free sequent calculus, but that have a complete calculus including restricted cut rules.\looseness=-1

Restricted cut rules circumvent the problem of the general cut rule for Maehara's method where one has no grip on the propositional variables occurring in the cut formula. We present three types of restricted cuts.
\begin{itemize}
    \item \emph{Analytic cut}: The \emph{analytic cut rule} $(\ancut)$ is presented in \cref{fig:cutrules}. The cut formula is restricted to subformulas of  the conclusion sequent. 
    The rule preserves the \emph{global subformula property} stating that in any proof each formula occurring in it is a subformula of the endsequent. Calculi with the subformula property are often called \emph{analytic}, justifying the name of the rule. Examples of Maehara's method with the analytic cut rule are for modal logic~$\Svijf$~\cite{Ono98} and bi-intuitionistic logic~\cite{KowalskiOno17}.
    \item \emph{Semi-analytic cut rules}: \emph{semi-analyticity} is the concept in which the propositional variables of the cut formula occur in \emph{one} formula from the conclusion. In other words, the cut formula is any (arbitrarily complex) combination of subformulas of a formula from the conclusion, justifying the name \emph{semi-analytic}. 
    An example is $(\semicut)$ from \cref{fig:cutrules} introduced in~\cite{Fitting83} in a tableau calculus. In \cref{sec:semi-analytic} we discuss a slightly different notion of semi-analytic rules, which provides a general setting in which Maehara's method can be applied.\looseness=-1
    \item \emph{Monochromatic cut} (e.g.,~\cite{Carbone97}): this type of cut is more general than semi-analytic cuts where the cut formula may be any combination of subformulas of \emph{multiple} formulas in the conclusion. One has to be careful which \emph{multiple} formulas are allowed and, therefore, this type of cut is only defined for split sequents: an application of the cut rule is \emph{monochromatic} in split sequent proof~$\pi$ with endsequent $\Gamma_1;\Gamma_2 \seqar \Delta_1 ; \Delta_2$ if{f} the cut formula~$\phi$ satisfies the following condition: $\Voc{\phi} \subseteq \Voc{\Gamma_1} \cup \Voc{\Delta_1}$ or $\Voc{\phi} \subseteq \Voc{\Gamma_2} \cup \Voc{\Delta_2}$, i.e.,  tracing the propositional variables only through the left or right side of Maehara's split.
\end{itemize}

\begin{figure}[t]
  \centering
  \begin{minipage}[t]{0.62\textwidth}
    \small{
\fbox{\parbox{0.95\textwidth}{
\centering{
\[
\vliinf{\ancut}{\phi \in \Sub{\Gamma,\Gamma',\Delta,\Delta'}}{\Gamma, \Gamma'\seqar\Delta,\Delta'}{\Gamma\seqar\Delta,\phi}{\phi,\Gamma'\seqar\Delta'}
\]
\[
\vliinf{\semicut}{
\psi \in \Sub{\Gamma,\Gamma',\Delta,\Delta'}}{\Gamma, \Gamma'\seqar\Delta,\Delta'}{\Gamma\seqar\Delta,\Box^n \psi}{\Box^n \psi,\Gamma'\seqar\Delta'}
\]
}
}
}}
\caption{Examples of restricted cut rules.}
\label{fig:cutrules}
  \end{minipage}
  \hfill
  \begin{minipage}[t]{0.35\textwidth}
    \small{
\fbox{\parbox{0.95\textwidth}{
\centering{
\[
\vlinf{\Tseqrule}{}{\Box \phi, \Gamma \seqar \Delta}{\Box \phi, \phi, \Gamma \seqar \Delta}
\]
\[
\vlinf{\Fiveseqrule}{}{\Sigma, \Box \Gamma \seqar \Box \phi, \Box \Delta}{\Box \Gamma \seqar \phi, \Box \Delta}
\]
}
}}}
\caption{Modal rules for~$\GSvijf$: $\GSvijf \ce \LK + (\ancut) + (\Tseqrule) + (\Fiveseqrule)$.}
\label{fig:GS5}
  \end{minipage}
\end{figure}

To illustrate the use of such cuts in Maehara's method, we consider modal logic~$\Svijf$ as an example from~\cite{Ono98}. $\Svijf$~is not known to have a cut-free sequent calculus. Sequent calculus~$\GSvijf \ce  \LK + (\ancut) + (\Tseqrule) + (\Fiveseqrule)$, including the analytic cut rule, is complete for~$\Svijf$ (cf.~\cref{fig:G3pc,fig:cutrules,fig:GS5} for all rules).

\begin{theorem}\label{thm:interpolationSvijf}
Modal logic~$\Svijf$ has the~CIP.
\end{theorem}
\begin{proof}
 The proof can be found in \cite{Ono98} and follows a Maehara-style argument such as in the proof of \cref{thm:Maehara_CPC}. We only state the interpolant constructions for the analytic cut rules and leave all other details to the reader. The interpolant construction for $(\ancut)$ depends on the formula in the conclusion whose subformula the cut formula~$\phi$ is. There are two cases:
\[
{\small{       \vliinf{L(\ancut)}{\phi \in \Sub{\Gamma_1, \Gamma_2, \Delta_1, \Delta_2}}{\Gamma_1, \Gamma_2; \Gamma'_1, \Gamma_2'\xseqar{\theta_1 \vee \theta_2}\Delta_1, \Delta_2; \Delta_1', \Delta_2'}{
            \Gamma_1;\Gamma_1'  \xseqar{\theta_1} \Delta_1, \phi; \Delta_1'
        }{
            \phi,\Gamma_2; \Gamma_2' \xseqar{\theta_2} \Delta_2; \Delta_2'
        }
        }}
\]
\[
{\small{
        \vliinf{R(\ancut)}{\phi \in \Sub{\Gamma_1', \Gamma_2', \Delta_1', \Delta_2'}}{\Gamma_1, \Gamma_2; \Gamma'_1, \Gamma_2'\xseqar{\theta_1 \wedge \theta_2}\Delta_1, \Delta_2; \Delta_1', \Delta_2'}{
            \Gamma_1; \Gamma_1' \xseqar{\theta_1} \Delta_1; \Delta_1', \phi
        }{
            \Gamma_2; \Gamma_2', \phi \xseqar{\theta_2} \Delta_2; \Delta_2'
        }
}}
\]
The variable condition~\ref{cond:seq_lang} from \cref{def:Maehara} is fulfilled  because of the subformula requirements.
\end{proof}
In fact, $\Svijf$~also admits the~LIP.  This cannot be shown using sequent calculus~$\GSvijf$, but can be obtained by using a hypersequent calculus for it (see~\cref{sec:hyper}).

\subsection{Uniform Interpolation via Pitts' Method} 
\label{subsec:Pitts}

The~UIP is a beautiful strengthening of the~CIP (see also  \refchapter{chapter:uniform}, the main chapter on uniform interpolation). Here, we are interested in Pitts'~\cite{Pitts92} important result that \emph{propositional quantification} can be interpreted in intuitionistic propositional logic $\IPC$ (i.e., only using the propositional language). To show this, Pitts provides a proof-theoretic method showing that the~UIP holds for~$\IPC$. The method was later adapted to (intuitionistic) modal logics~\cite{Bilkova06,Iemhoff19AML,AkbarTabatabaiJ18uniform}. The purpose of this section is to explain the main concepts in Pitts' method.

A Craig interpolant of $\phi \impl \psi$ is \textit{uniform} if it \textit{only} relies on~$\phi$ (or~only on~$\psi$) meaning that it is a Craig interpolant that uniformly works for any provable implication $\phi \impl \psi'$ with $\psi'$~in a specified vocabulary. 
Formally, we define the following, where a formula~$\phi$ is \emph{$p$-free} if{f} $p \notin \Voc{\phi}$, \emph{$p^+$-free} if{f} $p \notin \Vocpos{\phi}$, and \emph{$p^-$-free} if{f} $p \notin \Vocneg{\phi}$:

\begin{definition}[Uniform interpolation]\label{def:uniforminterpolation}
Fix a logic $\lgc$ over a language $\calL$. We say that $\lgc$ has the 
\emph{Uniform Interpolation Property~(UIP)} if{f} for any formula $\phi \in \form_{\calL}$ and propositional variable $p \in \Prop$, there are formulas $\exists p \phi, \forall p \phi$ over the same language~$\calL$ such that the following hold:\
\begin{enumerate}
\item
$\exists p \phi$ and $\forall p \phi$ are $p$-free, and $\Voc{\exists p \phi} \subseteq \Voc{\phi}$ and $\Voc{\forall p \phi} \subseteq \Voc{\phi}$;

\item 
$\Lproves \phi \impl \exists p \phi$ and $\Lproves \forall p \phi \impl \phi$;

\item
For all $p$-free formulas $\psi$:
\vspace{-0.5em}
\begin{align*}
    &\text{if } \Lproves \phi \impl \psi,
    \text{ then }  \Lproves \exists p \phi \impl \psi.\\
    &\text{if } \Lproves \psi \impl \phi,
    \text{ then }  \Lproves \psi \impl \forall p \phi.
\end{align*}
\end{enumerate}
Formula~$\exists p \phi$ is called the \emph{uniform post-interpolant} of~$\phi$ w.r.t.~$p$, and $\forall p \phi$~is the \emph{uniform pre-interpolant} of~$\phi$ w.r.t.~$p$ in~$\lgc$. When $\exists p \phi$ (resp.~$\forall p \phi$) exists for all formulas $\phi \in \form_{\calL}$ and $p \in \Prop$, then we say that $\lgc$~has the \emph{Post-Uniform Interpolation Property~(PostUIP)} (resp.~\emph{Pre-Uniform Interpolation Property~(PreUIP)}). 
\end{definition}

The notations~$\exists p \phi$~and~$\forall p \phi$ are suggestive: the formulas do not really contain quantifiers, but they are defined over the (propositional/modal) language~$\calL$. However, the notation is justified by Pitts' result stating that uniform interpolants \emph{interpret propositional quantifiers} of second order intuitionistic logic into the propositional language. 
    Uniform interpolants are also connected to so-called \emph{bisimulation quantifiers} (see  \refchapter{chapter:uniform} of this book). In the literature, and particularly in  \refchapter{chapter:propositional}, post-uniform and pre-uniform interpolants are known as right-uniform and left-uniform interpolants respectively.

\begin{remark}
\label{remark:classical_dual_ui}
In classical logic, we can define $\exists p \phi \ce  \neg \forall p (\neg \phi)$. \lipicsEnd
\end{remark}

The proof of the following property can be found, e.g., in~\cite{vdGiessen22}.

\begin{theorem}[Uniqueness]
Let $\chi_1$~and~$\chi_2$~both be uniform pre-interpolants for~$\phi$ w.r.t.~$p$ in logic~$\lgc$. Then $\Lproves \chi_1 \impl \chi_2$ and $\Lproves \chi_2 \impl \chi_1$. The same holds for uniform post-interpolants. 
\end{theorem}

Recently, Kurahashi~\cite{Kurahashi20} introduced the notion of uniform Lyndon interpolation:

\begin{definition}[Uniform Lyndon interpolation]
    A logic~$\lgc$ over a language~$\calL$ has the \emph{Uniform Lyndon Interpolation Property~(ULIP)}, if{f} for any $\phi \in \form_\calL$ and $p \in \Prop$, there are formulas~$\forallpos\! p\,  \phi$, $\forallneg\! p\,  \phi$, $\existspos\! p\,  \phi$,~and~$\existsneg\! p\,  \phi$ in~$\form_\calL$ that satisfy the properties from $\cref{def:uniforminterpolation}$, except that each occurrence of $p$-free is replaced with $p^+$-free for~$\forallpos\! p\,  \phi$~and~$\existspos\! p\,  \phi$ and with $p^-$-free for~$\forallneg\! p\,  \phi$~and~$\existsneg\! p\,  \phi$.
\end{definition}

\begin{theorem}[\cite{Kurahashi20}]
If a logic~$\lgc$ has the~UIP/ULIP, then it has the~CIP/LIP. 
\end{theorem}

In fact, uniform interpolants are the \textit{weakest} and \textit{strongest} Craig interpolants.

Now we turn to the proof-theoretic method by Pitts.
The ingredients for Pitts' method in comparison to Maehara's method are summarized in \cref{table:overview}. 
\begin{table}
\begin{align*}
&\text{\textbf{Maehara's method}}  && \text{\textbf{Pitts' method}}  \\ 
\hline \rule{0pt}{1\normalbaselineskip}
&\text{Craig interpolation} && \text{Uniform interpolation}\\
&\text{Sequent formulation of the Craig IP} \ \ \ \ \ \ \ && \text{Sequent formulation of the uniform IP} \ \ \ \ \ \ \ \\
&\text{`Cut-free' sequent calculus} && \text{`Strongly terminating' sequent calculus}\\
&\text{Proof of $\phi \impl \psi$} && \text{Proof search of $\phi$}
\end{align*}
\caption{General overview of similarities and differences between Maehara's and Pitts' method. The term `cut-free' is in quotes, because cut-free systems are nice for Maehara's method, but certainly not necessary (recall \cref{sect:role_of_cut}). Similarly, for Pitts' method, strongly terminating calculi are a nice feature to prove the~UIP, but not necessary (see \Cref{remark:terminating}).\vspace{-1em}}
\label{table:overview}
\end{table}
Because of technical reasons we will not prove the~UIP for~$\IPC$, as in the original proof by Pitts~\cite{Pitts92}, but we will explain Pitts' method for modal logic~$\K$ as done by B\'\i{}lkov\'a in~\cite{Bilkova06}. Recall \cref{remark:classical_dual_ui}: we only focus on the construction of pre-interpolants~$\forall p \phi$.

Towards a sequent based formulation of the~UIP we have the following in mind.  A~pre-interpolant for~$\phi$ w.r.t.~$p$ works almost as a Craig interpolant for all implications $\psi \impl \phi$ where $\psi$~is $p$-free (almost, because it might violate the Craig variable condition). So $\phi$~on the right is fixed, but $\psi$~on the left is flexible. This idea will be present in the split sequents in the following definition, where $\Gamma \seqar \Delta$ is fixed on the right of the split, and the $p$-free contexts~$\Pi$~and~$\Sigma$ are flexible and always on the left side of the split.

\begin{definition}
\label{def:seq_uip}
A sequent calculus~$\SC$ over a language~$\calL$ has the \emph{Sequent Pre-Uniform Interpolation Property~(SPreUIP)}  if{f} for any sequent $\Gamma \seqar \Delta$ and $p \in \Prop$, there exists a formula $A_p(\Gamma \seqar \Delta) \in \form_\calL$ such that the following properties hold:
\begin{enumerate}
\item
$\Voc{A_p(\Gamma \seqar \Delta)} \subseteq \Voc{\Gamma, \Delta} \setminus \{ p \}$;

\item
$\SCproves \Gamma, A_p(\Gamma \seqar \Delta) \seqar \Delta$;

\item
\label{item-three:sequent-uip}
$
\SCproves \Pi \seqar A_p(\Gamma \seqar \Delta), \Sigma
$ for any $p$-free multisets $\Pi$ and $\Sigma$ such that $\SCproves \Pi , \Gamma \seqar \Sigma , \Delta$.
\end{enumerate}
\end{definition}

\begin{remark}
    We only focus on modal logic~$\K$ for which above sequent formulation of the~UIP is sufficient because of \cref{remark:classical_dual_ui}. When considering intuitionistic logics one additionally has to define a formula~$E_p(\Gamma)$ that satisfies dual-like properties (not exactly dual since the construction of~$E_p(\Gamma)$ depends on~$A_p(\Gamma \seqar \phi)$). Here we leave out all such details and refer the reader to~\cite{Pitts92,AkbarTabatabaiJ18uniform}.  \lipicsEnd
\end{remark}
\begin{remark}
Here we only define a sequent version for the~UIP. One can also provide a sequent version for the~ULIP, see~\cite{AkbarTabatabaiIJ22nonnormal}. \lipicsEnd
\end{remark}

\begin{theorem}
    Let $\SC$~be a sequent calculus for a logic~$\lgc$. If $\SC$~has the~SPreUIP, then $\lgc$~has the~PreUIP.
\end{theorem}
\begin{proof}
    Define $\forall p \phi \ce  A_p(\ \seqar \phi)$. We leave it to the reader to check the details.
\end{proof}

$A_p(\Gamma \seqar \Delta)$ is constructed using the sequent calculus. Whereas for Maehara's method we are given a \textit{provable} split sequent $\Gamma ; \Gamma' \seqar \Delta ; \Delta'$, here we are given a sequent $\Gamma \seqar \Delta$ that might be \textit{unprovable}. So we cannot look at a \textit{proof} of $\Gamma \seqar \Delta$ because there simply might not be one. The trick is to construct the uniform interpolant via a \emph{proof search} of $\Gamma \seqar \Delta$. A proof search of a sequent is a tree in which, bottom-up, rules of the sequent calculus are applied in an attempt to find a proof. A key feature that we want of this proof search is that it is \textit{finite}. This helps ensure that the algorithm constructing the uniform pre-interpolant terminates.
\looseness=-1

\begin{definition}[Strong termination]
\label{def:termination}
A sequent calculus is said to be \emph{strongly terminating} if{f} for each sequent, any order of bottom-up applications of the rules results in a  finite tree, i.e., it stops with leaves that are either axioms or unprovable sequents to which no rule can be applied.
\looseness=-1
\end{definition}

\begin{remark}
\label{remark:terminating}
We focus on strongly terminating calculi, but weaker forms of termination may also work, though they typically require additional machinery, e.g., for modal logic $\T$ \cite{Bilkova06},  for a hypersequent calculus to prove the UIP for $\Svijf$~\cite{vdGiessenJK24}, or  in cyclic proofs~\cite{Afshari_etal21}. \lipicsEnd
\end{remark}

Many sequent calculi are not strongly terminating. For example, sequent calculus~$\LJ$ is not strongly terminating due to contraction rule~$(\lr\cntr)$  (see \cref{fig:interpolation}), but even when the explicit weakening and contraction rules are subsumed as in~$\Gthreepi$ from~\cite{TroelstraS00}, strong termination still fails, due to the intuitionistic version of rule~$(\lr\impl)$ where contraction is embedded as follows:
\[
\vliinf{(\lr\impl)^*}{}{\phi \impl \psi, \Gamma \seqar \Delta}{\phi \impl \psi, \Gamma \seqar \phi}{\psi, \Gamma \seqar \Delta}
\]
where $\Delta$ has at most one formula. Consequently, Pitts used a strongly terminating calculus~$\Gfourpi$ for~$\IPC$ where $(\lr\impl)^*$ is replaced with four other implication rules~\cite{Dyckhoff92,Hudelmaier93,Vorobev52,Vorobev70}, an idea later extended to intuitionistic modal logics~\cite{Iemhoff22}.

We look at modal logic~$\K$. Sequent calculus~$\GthreeK$ is not strongly terminating because of the contraction rules. Instead, we will consider calculus~$\GdrieK$ from \cref{fig:G3K} with implicit weakening and contraction.

\begin{figure}[t]
\centering{
\small{
\fbox{\parbox{0.97\textwidth}{
\vspace{-1em}
\[
\vlinf{(\id)^*}{}{p, \Gamma \seqar \Delta, p}{} \quad \quad
\vlinf{(\Kseqrule)^*}{}{\Pi, \Box \Gamma \seqar \Delta, \Box \phi}{\Gamma \seqar \phi}
\]
\[
\vlinf{(\lr\bot)^*}{}{\bot, \Gamma \seqar \Delta}{}\quad \quad
\vliinf{\lr\impl}{}{\phi \impl \psi,\Gamma\seqar\Delta}{\Gamma\seqar\Delta,\phi}{\psi, \Gamma \seqar \Delta} \quad \quad
\vliinf{\lr\vee}{}{\phi\vee\psi,\Gamma\seqar\Delta}{\phi,\Gamma\seqar\Delta}{\psi,\Gamma\seqar\Delta} \quad \quad
\vlinf{\lr\wedge}{}{\phi \wedge\psi,\Gamma\seqar\Delta}{\phi, \psi,\Gamma\seqar\Delta}
\]
\vspace{0.1em}
\[
\vlinf{(\rr\top)^*}{}{\Gamma\seqar\Delta,\top}{}\quad \quad
\vlinf{\rr\impl}{}{\Gamma\seqar\Delta,\phi \impl \psi}{\phi,\Gamma\seqar\Delta, \psi} \quad \quad \quad
\vlinf{\rr\vee}{}{\Gamma\seqar\Delta,\phi \vee \psi}{\Gamma\seqar\Delta,\phi,\psi} \quad \quad
\vliinf{\rr\wedge}{}{\Gamma\seqar\Delta,\phi\wedge\psi}{\Gamma\seqar\Delta,\phi}{\Gamma\seqar\Delta,\psi}
\]
\vspace{-0.3em}
}
}
}}
\caption{Sequent calculus~$\GdrieK$.}
\label{fig:G3K}
\end{figure}

\begin{theorem}
\label{thm:termination}
Sequent calculus~$\GdrieK$ for~$\K$ is strongly terminating.
\end{theorem}
\begin{proof}
    We define the weight~$w(\phi)$ of a formula~$\phi$ to be the number of symbols occurring in it counted as a multiset. For sequents, we define $w(\Gamma \seqar \Delta) \ce \Sigma_{\phi \in \Gamma} w(\phi)$. For all rules in~$\GdrieK$ it holds that the weight of each premise is smaller than the weight of the conclusion, meaning that the calculus is strongly terminating.
\end{proof}

\begin{remark}
Modal sequent rules~$(\Fourseqrule)$~and~$(\Tseqrule)$ from \cref{fig:modal-rules} prevent strong termination, because rule~$(\Fourseqrule)$ duplicates the context~$\Gamma$ and rule~$(\Tseqrule)$ duplicates principal formula~$\phi$ in the premise. In fact, logics~$\Kvier$~and~$\Svier$ do not have the~UIP~\cite{Bilkova06,GhilardiZ95}. 
\lipicsEnd
\end{remark}

\begin{theorem}\label{thm:Kuipsequent}
Sequent calculus~$\GdrieK$ has the~SPreUIP, and hence logic~$\K$ enjoys the~UIP.
\end{theorem}
\begin{proof}[Proof idea]
By recursion on $\Gamma \seqar \Delta$ we construct the formula $A_p(\Gamma \seqar \Delta)$ and prove that it acts as the pre-interpolant of $\Gamma \seqar \Delta$ w.r.t.~$p$. 
The algorithm starts a proof search bottom-up for sequent $\Gamma \seqar \Delta$ and inductively defines the uniform interpolants. The algorithm follows the steps from \cref{fig:uip} according to the shape of the sequents with priority given to the higher rows in the table. The left column of the table indicates the rule applied in the proof search, the middle column the shape of the conclusion of that rule, and the right column defines the pre-interpolant using the premise(s) of that rule. Rows~2--7 correspond to so-called \emph{invertible rules}. We stress that for each invertible rule, the pre-interpolant for the conclusion can always be taken to be a conjunction of the pre-interpolants of the premises. The last row is a special case corresponding to the \emph{non-invertible} rule~$(\Kseqrule)$ providing the $\Box$-clauses in the result, but it also corresponds to the possibility in which the proof search `applies rule~$(\Kseqrule)$' on a boxed formula of the extra context~$\Sigma$ on the `left split' from condition~\ref{item-three:sequent-uip} of \cref{def:seq_uip} providing the $\Diamond$-clause in the result. The algorithm terminates in the same ordering on weights~$w$ in which sequent calculus~$\GdrieK$ terminates (\cref{thm:termination}). Indeed, all rows except for the latter reflect the rules of the calculus, and also the recursion in the last row reduces in the same way. \looseness=-1

\begin{table}[t]
{\small{
\begin{align*}
& \textit{rule} &&\Gamma \seqar \Delta \textit{ matches} && A_p(\Gamma \seqar \Delta)\textit{ equals}\\
\hline \rule{0pt}{1\normalbaselineskip}
& (\id)^*  && r,\Gamma' \seqar \Delta', r && \top \\
& (\rr\top)^*  && \Gamma \seqar \Delta', \top && \top \\
& (\lr\bot)^*  && \bot,\Gamma' \seqar \Delta && \top \\
& (\lr\wedge) &&\Gamma',\psi_1 \wedge \psi_2 \seqar \Delta && A_p(\Gamma', \psi_1,\psi_2 \seqar \Delta) \\
& (\lr\vee) &&\Gamma',\psi_1 \vee \psi_2 \seqar \Delta && A_p(\Gamma',\psi_1 \seqar \Delta) \wedge A_p(\Gamma,\psi_2 \seqar \Delta) \\
& (\lr\impl) &&\Gamma', \psi_1 \impl \psi_2 \seqar \Delta && A_p(\Gamma' \seqar \Delta, \psi_1) \wedge A_p(\Gamma,\psi_2 \seqar \Delta) \\
& (\rr\wedge) &&\Gamma \seqar \Delta',\psi_1 \wedge \psi_2 && A_p(\Gamma \seqar \Delta', \psi_1) \wedge A_p(\Gamma \seqar \Delta', \psi_2)\\
& (\rr\vee) &&\Gamma \seqar \Delta',\psi_1 \vee \psi_2 && A_p(\Gamma \seqar \Delta', \psi_1,\psi_2)\\
& (\rr\impl) &&\Gamma \seqar \Delta',\psi_1 \impl \psi_2 && A_p(\Gamma, \psi_1 \seqar \Delta',\psi_2)\\
& (\Kseqrule)^* \text{ or no} && \Box \Gamma', \Phi \seqar \Box \Delta',\Psi \  && \Diamond A_p(\Gamma' \seqar \emptyset) \vee \bigvee_{q \in \Phi\setminus\{p\}} q  \vee \bigvee_{r \in \Psi \setminus \{p\}} \neg r \vee \bigvee_{\psi \in \Delta'} \Box A_p(\Gamma' \seqar \psi)    
\\[-1.5em]
& \text{rule applies} && \text{where } \Gamma' \neq \emptyset &&
\\
& (\Kseqrule)^* \text{ or no} && \Phi \seqar \Box \Delta',\Psi \  && \bigvee_{q \in \Phi\setminus\{p\}} q  \vee \bigvee_{r \in \Psi \setminus \{p\}} \neg r \vee \bigvee_{\psi \in \Delta'} \Box A_p(\seqar \psi)\\[-1.5em]
& \text{rule applies} &&  &&
\end{align*}
\vspace{-2em}
}}
\caption{Uniform interpolant construction for~$\GdrieK$, where $\Phi$ and $\Psi$ do not contain boxed formulas.\vspace{-1em}}
\label{fig:uip}
\end{table}

It remains to check the correctness of this algorithm, i.e., we need to check the conditions of \cref{def:seq_uip}. This is based on syntactic arguments for which we refer to~\cite{Bilkova06}. 
\end{proof}

\begin{example}
Let us apply the uniform interpolant algorithm for~$\GdrieK$ to formula $\Box p \vee \Box \neg p$. We present it in a proof search where we compute~$A_p(\cdot)$ of each sequent denoted in parentheses:
\[
\small{
\vlderivation{
\vlin{\rr\vee}{}{\ \seqar \Box p \vee\Box \neg p \ \ (\Box \bot \vee \Box \bot)}{
    \vliid{}{\text{Proof search on rule ($\Kseqrule$)}}{\ \seqar \Box p, \Box \neg p  \ \ (\Box \bot \vee \Box \bot)}{
    		\vlhy{\ \seqar p \ \ (\bot)}
    }
    {
    		\vlin{\rr\impl}{}{\ \seqar \neg p  \ \ (\bot)}{
        		\vlhy{p \seqar \bot \ \ (\bot)}
			}
    }
}
}
}
\]
For the top sequents, no rule of $\GdrieK$ applies, and the last row of \cref{fig:uip} applies with the empty disjunction yielding $\bot$. We leave it to the reader to check the conditions of \cref{def:seq_uip}. \looseness=-1\lipicsEnd
\end{example}

Similar proof-theoretic proofs of uniform interpolation have been provided for classical modal logics~$\GL$, $\SvierGrz$, $\T$~\cite{Bilkova06}, and~$\D$~\cite{Iemhoff19AML}, as well as for intuitionistic modal logics~$\iK$,~$\iKD$~\cite{Iemhoff19APAL}, and~$\iSL$~\cite{Feree_etal24}. The method has been extended to general forms of calculi providing results for a wide range of intermediate and substructural logics as discussed in the next section.

The uniform interpolation algorithm has been formalized in the interactive theorem prover Coq/Rocq for logics~$\IPC$, $\K$, $\GL$,~and~$\iSL$~\cite{Feree_vGool23,Feree_etal24}. This has resulted in an online calculator of uniform interpolants~\cite{Feree_vGool23,Feree_etal24}. This calculator has been useful to prove that certain second-order intuitionistic connectives are non-definable~\cite{Kocsis25}.

\section{Interpolation in Universal Proof Theory}
\label{sec:universal-proof-theory}

\emph{Universal proof theory}, initiated by Iemhoff \cite{Iemhoff19APAL,Iemhoff19AML} and extended in later works \cite{AkbarTabatabaiIJ21,AkbarTabatabaiIJ22int,AkbarTabatabaiIJ22nonnormal,AkbarTabatabaiJ18uniform,AkbarTabatabaiJ25,DP}, investigates general questions about the existence, equivalence, and characterization of well-behaved proof systems across broad classes of logics. It offers a uniform and modular method to establish logical properties, such as the~CIP~and~UIP and the admissibility of Visser rules~\cite{DP}, by linking them to the existence of certain structurally constrained sequent calculi. This framework applies to many logics, notably intuitionistic, modal, and substructural ones, by defining syntactic conditions that guarantee these properties.

The method yields both positive and negative results. Positively, if a logic admits a (terminating) \emph{semi-analytic} calculus satisfying the required structural conditions, it enjoys the~CIP~(and~UIP). In \cref{sec:semi-analytic} we introduce these calculi and see that they provide a general form for axioms and rules, such that as a result, the framework provides a systematic way to prove interpolation properties. Negatively, and more significantly, since these properties are rare, these results indicate that most logics cannot have a semi-analytic
(resp. terminating semi-analytic) calculus which we will discuss in \cref{sect:results}. In this sense, it not only provides a systematic way to prove interpolation results but also demonstrates the non-existence of well-behaved calculi for large classes of logics. 

Although interpolation properties were the first to be investigated within this line of research, \cite{DP} introduces a new logical invariant: the admissibility of the Visser rules. This extension significantly broadens the scope of universal proof theory, moving beyond interpolation to encompass further proof-theoretic properties. Such a generalization is particularly valuable for obtaining negative results, as it provides new tools to capture the limitations of proof systems. In \cite{DP}, a constructive sequent calculus is introduced, and one interesting result is that no intermediate logic other than $ \IPC$ admits a constructive sequent calculus.
\subsection{Semi-analytic Rules and Interpolation}
\label{sec:semi-analytic}
The following definition introduces the so-called \emph{semi-analytic} rules, a concept that is slightly weaker than the analyticity property of rules: a rule is \emph{analytic} when the formulas in the premises are subformulas of the principal formula in the conclusion. As a result, these rules are broader and more general than the analytic ones. 

\begin{definition}[Semi-analytic rules, \cite{AkbarTabatabaiJ18uniform,AkbarTabatabaiJ25}] \label{semi-analyticRules}
For $1 \leq i \leq n$ and $1 \leq j \leq m$, let $\Gamma_i$'s, $\Delta_i$'s, and~$\Pi_j$'s be pairwise disjoint families of pairwise distinct multiset variables, $\overline{\phi}_{ir}$'s, $\overline{\psi}_{js}$'s, and~$\overline{\chi}_{js}$'s be multisets of formulas and $\phi$ a formula, where $1 \leq r \leq k_i$ and $1 \leq s \leq l_j$.  

A rule is called \emph{left single-conclusion semi-analytic} if it has the form
\[
\vliinf{}{}{\Pi_1, \dots, \Pi_m, \Gamma_1, \dots, \Gamma_n, \phi \seqar \Delta_1, \dots, \Delta_n}{\{\Pi_j , \overline{\psi}_{js} \seqar \overline{\chi}_{js} \mid 1 \leq j \leq m, 1 \leq s \leq l_j \}}{\quad \{\Gamma_i , \overline{\phi}_{ir} \seqar \Delta_i \mid 1 \leq i \leq n, 1 \leq r \leq k_i \}}
\]
and it satisfies the following variable condition:
\[
\bigcup_{ir} \Voc{\overline{\phi}_{ir}} \cup \bigcup_{js} \Voc{\overline{\psi}_{js}} \cup \bigcup_{js} \Voc{\overline{\chi}_{js}}
  \subseteq \Voc{\phi}.   
\]
In each instance of the rule, at most one of $\Delta_i$'s can be substituted by a formula and the rest must be empty, as the rule is single-conclusion. 

A rule is called \emph{right single-conclusion semi-analytic} if it has the form
\[
\vlinf{}{}{\Gamma_1, \dots, \Gamma_n \seqar \phi}{\{\Gamma_i , \overline{\phi}_{ir} \seqar \overline{\chi}_{ir} \mid 1 \leq i \leq n, 1 \leq r \leq k_i\}}
\]
and the following variable condition holds:
\[
     \bigcup_{ir} \Voc{\overline{\phi}_{ir}} \cup \bigcup_{ir} \Voc{\overline{\chi}_{ir}} \subseteq \Voc{\phi}.
\]
A \emph{multi-conclusion semi-analytic} rule can be defined similarly, when we allow the rule to be multi-conclusion (see~\cite[Definition 4.3]{AkbarTabatabaiJ25}). A rule is called \emph{semi-analytic} if it is a single-conclusion or multi-conclusion left or right semi-analytic rule.
\end{definition}

\begin{example}\label{ExampleSemi-analyticRules}
A generic example of a left single-conclusion semi-analytic rule is the following:
\[
\vliiinf{}{}{\Pi, \Gamma, \mu \seqar \Delta}{\Pi, \nu^1_{11}, \nu^2_{11} \seqar \rho_{11}}{\Pi, \nu_{12} \seqar \rho_{12}}{\Gamma, \mu^1_1, \mu^2_1, \mu^3_1 \seqar \Delta}
\]
where 
$
\Voc{\mu^1_1} \cup \Voc{\mu^2_1} \cup \Voc{\mu_1^3} \cup \Voc{\nu_{11}^1} \cup \Voc{\nu_{11}^2} \cup \Voc{\nu_{12}} \cup \Voc{\rho_{11}} \cup \Voc{\rho_{12}} \subseteq \Voc{\mu}$.

The left and middle premises are  in the same family with the context $\Pi$ in the antecedent. Thus, one copy of $\Pi$ appears in the antecedent of the conclusion.  \lipicsEnd
\end{example}

\begin{example}
Concrete examples of multi-conclusion semi-analytic rules are the rules in \cref{fig:G3pc}, except for the cut rule. The single-conclusion version of these rules are single-conclusion semi-analytic rules.
Concrete non-examples of multi-conclusion semi-analytic rules include the cut rule, as it violates the variable condition. Considering the modal rules in \cref{fig:modal-rules},  rules~$(\Tseqrule)$~and~$(\Sfourseqrule)$ are single-conclusion semi-analytic, $(\Kseqrule)$~and~$(\Dseqrule)$~are non-examples as the contexts do not remain intact from the premises to the conclusion, and $(\Fourseqrule)$~and~$(\GLseqrule)$ are also non-examples, as context~$\Gamma$ disappears from the premise to the conclusion.
Non-examples for single-conclusion semi-analytic rules are provided by the single-conclusion version of the mentioned rules, i.e.,~$(\Kseqrule)$, $(\Dseqrule)$, $(\Fourseqrule)$, $(\GLseqrule)$,~and~$(\cut)$. For another non-example, consider the following rule in the calculus~$\mathbf{KC}$ introduced in \cite{bor} for the logic $\KC= \IPC+\neg \phi \vee \neg \neg \phi$:
\[
\vlinf{}{}{\Gamma \seqar \phi \impl\psi, \Delta}{\Gamma, \phi \seqar \psi, \Delta}
\qquad
\vliinf{}{}{\Gamma, \phi \impl\psi, \Pi \seqar  \Delta, \Lambda}{\Gamma \seqar \Delta, \phi}{\ \ \psi, \Pi \seqar \Lambda}
\]
in which $\Delta$ must be substituted by multisets of \emph{strictly negative formulas} (not containing all formulas in the language; see \cite{bor}).  
This rule is not semi-analytic as its context $\Delta$ is not free for arbitrary substitutions of multisets. \lipicsEnd
\end{example}

\begin{remark}
Note that the semi-analytic cut rules introduced in \cref{subsec:Maehara} do not fall in the category of semi-analytic rules as per \cref{semi-analyticRules}. The reason is the simple fact that in a cut rule, there is no principal formula. However, it is possible to generalize the notion to also cover semi-analytic cut rules, a direction we leave for future work.
\lipicsEnd
\end{remark}

 Now, we define a general form for axioms.

\begin{definition}[Focused axioms, \cite{AkbarTabatabaiJ25}]\label{Dfn Focused Axioms}
A meta-sequent is a \emph{focused axiom}, if it has one of the forms:
\[
 \phi \seqar \phi \quad \quad  \quad \quad \seqar \bar{\psi} \quad \quad  \quad \quad  \bar{\theta} \seqar \quad \quad  \quad \quad \Gamma, \bar{\phi} \seqar \Delta \quad \quad  \quad \quad  \Gamma \seqar \bar{\phi}, \Delta
\]
where for each axiom $\Voc{\mu}= \Voc{\nu}$, for any $\mu, \nu \in \bar{\eta}$ and $\eta \in \{\phi, \psi, \theta\}$.  
\end{definition}

Note that the term \emph{focused} is different from the notion of `focused proofs' used in proof-search, as in the works \cite{Dale, Dale1}.
\begin{example}
The axioms in sequent calculi~$\LK$~and~$\LJ$ are all focused. More examples:
\[
p, \neg p \seqar \quad \quad  \quad \quad  \seqar p, \neg p \quad \quad  \quad \quad  \Gamma, \neg \top \seqar \Delta \quad \quad  \quad \quad \Gamma \seqar \Delta, \neg \bot   
\]
For a non-example, consider the axiom $(p, \neg p, q \seqar \,)$, where $p$ and $q$ are distinct atoms. This axiom is not focused, as the set of the variables of $p$ and $q$ 
are not equal. \lipicsEnd
\end{example}

\begin{definition}
A calculus~$\SC$ is \emph{(single-conclusion) multi-conclusion semi-analytic} if each rule in~$\SC$ is (single-conclusion) multi-conclusion semi-analytic and its axioms are (single-conclusion versions of) focused axioms. In case the language contains the modality~$\Box$, any union of the sets of rules $\{(\Kseqrule)\}, \{(\Fourseqrule)\}, \{(\Kseqrule),(\Dseqrule)\}, \{(\Kseqrule),(\Tseqrule)\}$ are also allowed to appear in~$\SC$. \looseness=-1
\end{definition}

We finish this subsection by defining what it means for a calculus to be fully terminating. This is a stronger notion than \emph{strong termination} introduced in \cref{def:termination}.

\begin{definition}[Fully terminating sequent calculi]\label{Dfn:TerminatingCalculi}
Let $\SC$~be a calculus and $\preceq$~be a well-founded order on sequents. We say that $\SC$~is \emph{fully terminating with respect to~$\preceq$} 
if for any sequent~$S$ there are at most finitely many instances of the rules in~$\SC$ with the conclusion~$S$. Moreover, we require the following to be strictly less than~$S$ with respect to the order~$\preceq$:
\begin{itemize}
\item
the premises of all instances of a rule where the conclusion is $S$,

\item
proper subsequents of $S$, and

\item
in case that the language contains the modality $\Box$, any $S'=(\Gamma, \Pi \seqar \Delta, \Lambda)$, where $\Pi \cup \Lambda \neq \emptyset$ and $S=(\Gamma, \Box \Pi \seqar \Delta, \Box \Lambda)$.
\end{itemize}
The calculus $\SC$ is \emph{fully terminating} if there is an order $\preceq$ such that $\SC$ is fully terminating with respect to $\preceq$.
\end{definition}

\begin{remark}
   In modal logics, a fully terminating semi-analytic calculus cannot contain the modal rules $(\Fourseqrule)$ and $(\Tseqrule)$ presented in Figure~\ref{fig:modal-rules}. This means that any fully semi-analytic modal sequent calculus can contain the set of modal rules $\{(\Kseqrule)\}$ and $\{(\Kseqrule),(\Dseqrule)\}$. \lipicsEnd
\end{remark}

A useful rule of thumb for interpolation is that if a logic~$\lgc$ has a semi-analytic calculus~$\SC$, then $\lgc$~enjoys the~CIP, and if $\SC$~is fully terminating, it also enjoys the~UIP. The proofs essentially follow Maehara's approach for the~CIP and Pitts' method for the~UIP. Of course, there are a lot of technical details to check for which we refer to \cite{AkbarTabatabaiJ25,AkbarTabatabaiJ18uniform}.

\subsection{Interpolation and (Non-)Existence of Sequent Calculi}
\label{sect:results}
In this section, we apply the uniform methodology introduced in the context of universal proof theory and present a rich number of positive and negative results based on the~CIP~and~UIP across different logics.  
For a positive result we first observe that a given logic has a (fully terminating) semi-analytic calculus, which allows us to conclude that this logic has the~CIP~(UIP). For a negative result we use the contraposition: if we know that a logic does not have the~CIP~(UIP), we immediately conclude that this logic cannot have a (fully terminating) semi-analytic calculus. 

\subsubsection{Intermediate and Modal Logics}

The following theorem is the main positive result with regard to proving the~CIP~and~UIP.

\begin{theorem}[\cite{AkbarTabatabaiJ25,AkbarTabatabaiJ18uniform}]\label{MainCor} 
Let $\lgc$~be a logic over a language $\calL \in \{\calLp, \calLBox\}$ and $\SC$~be a single-conclusion (resp.~multi-conclusion) semi-analytic calculus for~$\lgc$. 
\begin{itemize}
    \item If $\calL = \calLp$ and if $\IPC \subseteq \lgc$ (resp.~$\lgc = \CPC$), then $\lgc$~has the~CIP. Moreover, if $\SC$~is fully terminating, then $\lgc$~has the~UIP.
    \item If $\calL = \calLBox$ and if $\K \subseteq \lgc$, then~$\lgc$ has the~CIP. Moreover, if $\SC$~is fully terminating, then~$\lgc$ has the~UIP.
\end{itemize}
\end{theorem}

The theorem in its current form does not yield strong general positive results for two reasons. First, in the realm of intermediate logic it is known that only seven consistent intermediate logics have the CIP and UIP, namely $\IPC$, $\LC$, $\KC$, $\BDtwo$, $\Sm$, $\GSc$,~and~$\CPC$~\cite{Max}. Second, for modal calculi, being semi-analytic depends on an ad-hoc examination of rules, as we currently lack a general form for a semi-analytic modal rule.
However, for modal logics there are many proof-theoretic proofs of the CIP, e.g., for modal logic~$\GL$~\cite{Bilkova06} and intuitionistic modal logics in~\cite{Iemhoff19APAL} (see also \Cref{sec:basics}).
 
To state our negative results (the following corollaries), we rely on the  nonexistence of the~CIP~and~UIP in a wide range of logics~\cite{Bil,GhilardiZ95,Ghil2,Max,max91}. Consequently, we get the nonexistence of (fully terminating) calculi consisting only of semi-analytic rules and focused axioms. \looseness=-1

\begin{corollary}[Negative results,~CIP]
The following results hold:
\begin{itemize}
    \item 
Except for logics~$\IPC$, $\LC$, $\KC$, $\BDtwo$, $\Sm$, $\GSc$,~and~$\CPC$, no intermediate logic has a single- or multi-conclusion semi-analytic sequent calculus.
\item 
With the possible exception of at most~37 for~$\Svier$ and 6~for~$\Grz$, none of their consistent extensions have a single- or multi-conclusion semi-analytic sequent calculus.
\end{itemize}
\end{corollary}

\begin{corollary}[Negative results,~UIP]
With the possible exception of at most six extensions of~$\Svier$, none of the extensions of either~$\Svier$~or~$\Kvier$ have a single- or multi-conclusion fully terminating semi-analytic sequent calculus. 
\end{corollary}

\subsubsection{(Modal) Substructural and Linear Logics} 
The results of the previous section are generalized to (modal) substructural and linear logics. Consider the \emph{bounded} language  $\calL^b\ce\{\wedge, \vee, \to, *, 0, 1, \top, \bot\}$ and the \emph{unbounded} language $\calL^{u}\ce\calL^b \setminus \{\top, \bot\}$. Denote $\calL_!^b\ce\calL^b \cup \{!\}$ and $\calL_!^u\ce\calL^u \cup \{!\}$. 
We recall some well-known calculi for substructural and linear rules in \cref{fig:FLe,fig:rules_bang}. The left column of \cref{fig:sequent_calculi_substr_etc} presents well-known sequent calculi using these rules, the weakening and contraction rules from \cref{fig:G3pc}, and $(\lr\bot)^*$ and $(\rr\top)^*$ from \cref{fig:G3K}. The single-conclusion versions of the calculi in the left column are~$\FLseqe$, $\FLseqew$, $\FLseqec$, $\IMALLseq$, $\AIMALLseq$, and $\RIMALLseq$ respectively.

\begin{figure}[t]
\centering{
\small{
\fbox{\parbox{0.97\textwidth}{
\[\vlinf{\id_\phi}{}{\phi \seqar \phi}{} \qquad \qquad 
\vlinf{\lr0}{}{0\seqar}{}\qquad \qquad\vlinf{\rr1}{}{\seqar1}{}\qquad \qquad \vlinf{\rr0}{}{\Gamma\seqar\Delta,0}{\Gamma\seqar\Delta}
\qquad \qquad 
\vlinf{\lr1}{}{1,\Gamma\seqar\Delta}{\Gamma\seqar\Delta}
\]
\vspace{0.3em}
\[
\vlinf{\lr* }{}{\phi*\psi,\Gamma\seqar\Delta}{\phi,\psi,\Gamma\seqar\Delta} \ \ \ 
\vliinf{(\lr\impl)^*}{}{\phi \impl \psi,\Gamma,\Gamma'\seqar\Delta,\Delta'}{\Gamma\seqar\Delta,\phi}{\psi, \Gamma' \seqar \Delta'} \ \  \
\vliinf{\lr\vee}{}{\phi\vee\psi,\Gamma\seqar\Delta}{\phi,\Gamma\seqar\Delta}{\psi,\Gamma\seqar\Delta} \ \ \ 
\vlinf{\lr\wedge^i }{\mathsmaller{(i=0,1)}}{\phi_0\wedge\phi_1,\Gamma\seqar\Delta}{\phi_i,\Gamma\seqar\Delta}
\]
\vspace{0.5em}
\[
\vliinf{\rr*}{}{\Gamma,\Gamma'\seqar\Delta,\Delta',\phi*\psi}{\Gamma\seqar\Delta,\phi}{\Gamma'\seqar\Delta',\psi} \ \ \ 
\vlinf{\rr\impl}{}{\Gamma\seqar\Delta,\phi \impl \psi}{\phi,\Gamma\seqar\Delta, \psi} \ \ \ 
\vlinf{\rr\vee^i}{\mathsmaller{(i=0,1)}}{\Gamma\seqar\Delta,\phi_0 \vee \phi_1}{\Gamma\seqar\Delta,\phi_i} \ \ \ 
\vliinf{\rr\wedge}{}{\Gamma\seqar\Delta,\phi\wedge\psi}{\Gamma\seqar\Delta,\phi}{\Gamma\seqar\Delta,\psi}
\]
}
}
}}
\caption{Sequent calculus~$\CFLseqe$ for substructural logic~$\CFLe$. Rules~$(\lr\vee)$, $(\lr\wedge^i)$, $(\rr\impl)$, $(\rr\vee^i)$, $(\rr\wedge)$,~and~$(\cut)$ are the same as in \cref{fig:G3pc}.}
\label{fig:FLe}
\end{figure}

\begin{figure}[t]
\centering{
\small{
\fbox{\parbox{0.97\textwidth}{
\vspace{-0.3em}
\[\vlinf{\rr!}{}{!\Gamma\seqar !\phi}{!\Gamma\seqar \phi} \quad \quad \vlinf{\lr!}{}{\Gamma,!\phi\seqar\Delta}{\Gamma,\phi\seqar\Delta} \quad \quad \vlinf{\wk!}{}{\Gamma,!\phi\seqar\Delta}{\Gamma\seqar\Delta}\quad \quad\vlinf{\cntr!}{}{\Gamma,!\phi\seqar\Delta}{\Gamma,!\phi,!\phi\seqar\Delta}\quad \quad \vlinf{\Kseqrule!}{}{!\Gamma\seqar!\phi}{\Gamma\seqar\phi}
\quad \quad 
\vlinf{\Dseqrule!}{}{!\Gamma\seqar}{\Gamma\seqar}
\]
\vspace{-0.3em}
}
}
}}
\caption{Sequent rules for `$!$' for linear logic.}
\label{fig:rules_bang}
\end{figure}

\begin{figure}[t]
\fbox{\parbox{0.97\textwidth}{
  \centering
   \hspace{-1.15cm}
    \begin{minipage}[t]{0.36\textwidth}
    \small{
    \vspace{0.5em}
    Calculi over $\calL^{u}$:
    \begin{align*}
        \CFLseqe &\ce  \text{\Cref{fig:FLe}}\\
        \CFLseqew &\ce  \CFLseqe + \{(\lr\wk),(\rr\wk)\}\\
        \CFLseqec &\ce  \CFLseqe + \{(\lr\cntr),(\rr\cntr)\}
    \end{align*}
    \vspace{-1.5em}
    }
  \end{minipage}
    \hspace{1cm}
      \begin{minipage}[t]{0.45\textwidth}
    \small{
    \vspace{0.40em}
    Calculi over $\calL^{b}_{!}$:
    \begin{align*}
        \CLLseq &\ce  \MALLseq + \{(\rr!),(\lr!),(\wk!),(\cntr!) \}\\
        \ILLseq &\ce  \IMALLseq + \{(\rr!),(\lr!),(\wk!),(\cntr!) \}\\
        \ELLseq &\ce  \MALLseq + \{(\wk!),(\cntr!),(\Kseqrule!),(\Dseqrule!) \}\\
        \ALLseq &\ce  \CLLseq + \{(\lr\wk),(\rr\wk) \}\\
        \RLLseq &\ce  \CLLseq + \{(\lr\cntr),(\rr\cntr) \}\\
    \end{align*}
    \vspace{-1.5em}
    }
  \end{minipage}
  \\
  \hspace{-8cm}
  \begin{minipage}[t]{0.41\textwidth}
    \small{
    \vspace{-4em}
    Calculi over $\calL^{b}$:
    \begin{align*}
        \MALLseq &\ce  \CFLseqe + \{(\lr\bot)^*,(\rr\top)^*\}\\
        \AMALLseq &\ce  \MALLseq + \{(\lr\wk),(\rr\wk)\}\\
        \RMALLseq &\ce  \MALLseq + \{(\lr\cntr),(\rr\cntr)\}
    \end{align*}
    \vspace{-1.5em}
    }
  \end{minipage}
  }}
    \caption{Sequent calculi for logics over different languages. In  $\ILLseq$ we use the single-conclusion version of the rules. Rules~$(\lr\wk)$, $(\rr\wk)$, $(\lr\cntr)$,~and~$(\rr\cntr)$ are displayed in \cref{fig:G3pc}, and rules $(\lr\bot)^*$ and $(\rr\top)^*$ in \cref{fig:G3K}.} \label{fig:sequent_calculi_substr_etc}
\end{figure}

For the language $\calL^{b}_{!}$ that contains `!', we present sequent calculi for several logics on the right column of \cref{fig:sequent_calculi_substr_etc}, using the rules in \cref{fig:rules_bang}.  It is worth noting that by denoting~$!$ in the language by $\Box$, the rules in \cref{fig:rules_bang} resemble standard modal rules. That is why we, e.g., write $\FLseqeK$ to mean $\FLseqe+\{\Kseqrule!\}$, write $\FLseqewD$ to mean $\FLseqew + \{ \Kseqrule!,\Dseqrule!\}$, or write $\FLseqeT$ to mean $\FLseqe + \{\lr!\}$, and likewise, we write $\FLe\sf{K}$, $\FLew\sf{D}$, and $\FLe\sf{T}$ for their corresponding logics. The following calculi are proved to enjoy cut admissibility in \cite{Dosenjadid,ono,Gentzen35a,LLL,affine,kanovich,kiriyama,Ono90,Ono98,Onokomori,tro92}, as well as \cite[Proposition 4]{dosen}:
\begin{theorem}\label{thm: cut mardom}
The following sequent calculi enjoy cut admissibility:
 $\FLseqe$, $\FLseqew$, $\FLseqec$, 
 $\CFLseqe$,   $\CFLseqew$,
 $\CFLseqec$,
 $\LJ$, 
 $\LK$,
 $\IMALLseq$, $\AIMALLseq$, $\RIMALLseq$, 
 $\MALLseq$,   $\AMALLseq$, $\RMALLseq$,
  $\mathbf{X}+\{(\rr!), (\lr!)\}$ where $\mathbf{X} \in \{\CFLseqe, \CFLseqew,  \CFLseqec, \LK\}$, $\CFLseqe+\{(\rr!), (\lr!), (\cntr!)\}$, $\CFLseqew+\{(\rr!), (\lr!), (\cntr!)\}$,
$\ILLseq$,  $\CLLseq$,
$\ALLseq$, $\RLLseq$, 
$\ELLseq$.
 \end{theorem}

Then, the analog of \cref{MainCor} for substructural and linear logics will become the following (the results can be generalized to multimodal substructural logics \cite{AkbarTabatabaiJ25}, which we do not cover due to space limitations):

\begin{theorem}[\cite{AkbarTabatabaiJ25}]\label{MainCorSub}
Let $\lgc$~be a logic over $\calL \in \{\calL^b, \calL^u, \calL_!^b, \calL_!^u\}$ and $\SC$~be a single-conclusion (resp.~multi-conclusion) semi-analytic calculus for~$\lgc$. Then, 
\begin{itemize}
\item 
if $\calL =\calL^{b}$ or $\calL =\calL_!^{b}$ and $\IMALL \subseteq \lgc$ (resp.~$\MALL \subseteq \lgc$),  then $\lgc$ has the~CIP, and
\item 
if $\calL =\calL^{u}$ or $\calL =\calL_!^{u}$ and $\FLe \subseteq \lgc$ (resp.~$\CFLe \subseteq \lgc$),  then $\lgc$ has the~CIP.
\end{itemize}
\end{theorem}

\begin{corollary}[Positive results]
\label{ApplicationInt1} 
Concrete examples of the logics that enjoy the~CIP are the logics of the sequent calculi in \cref{thm: cut mardom}.
\end{corollary}
\begin{proof}
The presented sequent calculi for these logics consist of focused axioms and semi-analytic rules. Therefore, by \cref{MainCor}, we get the~CIP for their logics.
\end{proof}

We use a series of facts regarding the nonexistence of the~CIP in a wide range of logics to obtain our negative results. Before that, we recall some of these logics. \Cref{fig:equational-conditions} presents equational conditions to define varieties of pointed commutative residuated lattices and \cref{fig:equational-logics} presents equational logics in the usual sense. Let \textbf{A} be a single pointed commutative residuated
lattice and \textbf{Fm} be the formula algebra for the language of pointed commutative residuated lattices. Given any variety $\mathbb{V}$ of pointed commutative residuated lattices, a \emph{logic}~$\lgc^{\mathbb{V}}$ may be defined by fixing, for all $\Gamma \cup\{\varphi\} \subseteq \mathbf{Fm}$,
\[
\begin{aligned}
\Gamma \vdash_{\lgc^{\mathbb{V}}} \varphi \Leftrightarrow & \text { for any } \mathbf{A} \in \mathbb{V} \text { and homomorphism } v\colon \mathbf{F m} \rightarrow \mathbf{A}: \\
& \text { if for all } \psi \in \Gamma, \quad 1 \leqslant v(\psi), \quad \text { then } \quad 1 \leqslant v(\phi)
\end{aligned}
\]
By $\mathcal{V}(\mathbb{K})$, we mean the variety generated by a class of pointed commutative residuated lattices~$\mathbb{K}$.
Let us define some fundamental fuzzy logics.  For any $n > 1$, let ${L}_n = \{0, \frac{1}{n-1}, \dots , \frac{n - 2}{n - 1}, 1 \}$ and
\[
\mathbb{L}_n = \langle L_n , min, max , *_{\text{\L}}, \to_{\text{\L}}, 1, 0 \rangle
 \quad , \quad
\mathbb{G}_n = \langle L_n , min, max , min, \to_{G}, 1, 0 \rangle   
\]
where $x *_{\text{\L}} y = max \{0, x +y-1\}$, $x \to_{\text{\L}} y = min \{1, 1-x+y\}$, and $x \to_G y$ is $y$ if $x > y$, otherwise 1. Define $\Lukas_n$ and $\Godel_n$ for $n >1$ as the logics of $\mathcal{V}(\mathbb{L}_n)$ and $\mathcal{V}(\mathbb{G}_n)$, respectively (see~\cite{March}). Separately, define the relevance logic $\Rlogic$ as the logic of the variety of distributive pointed commutative residuated lattices satisfying $x \leq x*x$ and $\neg \neg x=x$ for any $x$. 
\begin{figure}[!t]
  \centering
\fbox{\parbox{0.98\textwidth}{
{\centering
$ \ \ \ 
\begin{array}{ll}
    \prl & 1 \leq (x \impl y) \vee (y \impl x)\\
    \disaxiom & x \wedge (y \vee z) = (x \wedge y) \vee (x \wedge z)\\
    \inv & \neg \neg x = x \\
    \intaxiom & x \leq 1 \\
    \bd & 0 \leq x\\
    \idaxiom & x = x * x \\
\end{array}
\quad \ \ \ 
\begin{array}{ll}
    \fp & 0=1 \\
    \divaxiom & x * (x \impl y) = y * (y \impl x)\\
    \can & x \impl (x * y)= y \\
    \rcan & 1= \neg x \vee ((x \impl (x * y)) \impl y)\\
    \nc & x \wedge \neg x \leq 0 \\
     &  \\
\end{array}
$
}}}
\caption{Some equational conditions.} \label{fig:equational-conditions}
\end{figure}

\begin{figure}[!t]
    \small{
\fbox{\parbox{0.97\textwidth}{
\setlength\abovedisplayskip{0pt}
\setlength\belowdisplayskip{0pt}
\ \ 
\begin{align*}
    \ULm &\ce  \FLe  + \prl + \disaxiom   && & \Lukas &\ce  \MTL  + \divaxiom + \inv\\
    \IULm &\ce  \ULm  + \inv   && & \Plogic &\ce  \MTL  + \divaxiom + \rcan\\
    \MTL &\ce  \ULm  + \intaxiom + \bd     && & \CHL &\ce  \ULm  + \intaxiom + \fp + \divaxiom + \can\\
    \SMTL &\ce  \MTL  + \nc    && & \UMLm &\ce  \ULm  + \idaxiom\\
    \IMTL &\ce  \MTL  + \inv     && & \RMe &\ce  \UMLm  + \inv\\
    \BL &\ce  \MTL  + \divaxiom && & \IUMLm &\ce  \UMLm  + \inv + \fp\\
    \Godel &\ce  \MTL  + \idaxiom && & \Alogic &\ce  \IULm  + \fp + \can
\end{align*}
}
}}
\caption{Some substructural logics.  In each logic, axioms $\prl$ and $\disaxiom$ are present.}
\label{fig:equational-logics}
\end{figure}

\begin{corollary}[Negative results,~CIP] \label{Cor_Negative_Applications}
The following two families of logics lack the~CIP and hence cannot have a single-conclusion semi-analytic sequent calculus. The first family:
\begin{itemize}
\item {\sf{(\cite{ba,luk,Urq}\cite[Corollary 3.2]{March})\textbf{.}}}
 $\ULm$, $\MTL$, $\SMTL$, $\BL$, $\Plogic$, $\CHL$, $\Godel_n$ (for $n \geq 4$);
\item {\sf{(\cite{Max})\textbf{.}}}
All consistent extensions of~$\IPC$, except for possibly~$\IPC$, $\Godel$, $\KC$, $\BDtwo$, $\Sm$, $\GSc$,~and~$\CPC$; 
\item
{\sf{(\cite{ba,mo}\cite[Theorem 3.3]{March})\textbf{.}}}
All consistent $\BL$-extensions, except for possibly $\Godel$, $\mathsf{G3}$ and~$\CPC$;
\end{itemize}
And the second family:
\begin{itemize}
\item
{\sf{(\cite{ba,luk,Urq}\cite[Corollary 3.2]{March})\textbf{.}}}
$\Rlogic$, $\Alogic$, $\IULm$, $\IMTL$, $\Lukas$, $\Lukas_n$ (for $n \geq 3$);
\item
{\sf{(\cite[Corollary 3.5]{March}\cite{komori81})\textbf{.}}}
All consistent $\IMTL$-extensions, and consequently, all consistent extensions of
 $\Lukas$, except for $\CPC$;
\item
{\sf{(\cite[Theorem 4.9]{March})\textbf{.}}} All consistent extensions of $\RMe$, except for possibly $\RMe$, $\IUMLm$, $\CPC$, $\RMe_3$, $\RMe_4$, $\CPC \cap \IUMLm$, $\RMe_4 \cap \IUMLm$, and $\CPC \cap \RMe_3$. This includes:
$\RMe_n$ for $n \geq 5$,
$\RMe_{2m} \cap \RMe_{2n+1}$ for $n \geq m \geq 1$ with $n \geq 2$, and
$\RMe_{2m} \cap \IUMLm$ for $m \geq 3$;
\item
{\sf{(\cite[Corollary 3.14]{fussner})\textbf{.}}}
Continuum-many
extensions of $\MALL$ and $\CLL$.
\end{itemize}
The logics in the second family do not have a multi-conclusion semi-analytic calculus either.
\end{corollary}

Regarding the~UIP, we have the following result. Our theorem generalizes the results of~\cite{AlizadehDO14}~and~\cite{Bil} to also cover the modal cases:

\begin{corollary}[Positive results,~UIP]\label{FLeUniformInterpolation}
Logics $\CFLe$, $\CFLew$, and $\CPC$,  their $\K$ and $\D$ modal versions, as well as logics $\FLe$, $\FLew$, $\FLe\K$, $\FLew\K$, $\FLe\D$, $\FLew\D$, $\FLe\T$, and  $\FLew\T$ have the~UIP.\looseness=-1
\end{corollary}
\begin{proof}
The rules of the usual calculi of these logics are semi-analytic, and their axioms are focused, and in the absence of the contraction rule, the calculi are clearly fully terminating.
\end{proof}

\subsubsection{Non-normal Modal and Conditional Logics}

It is possible to extend the methodology to cover (non-normal) modal logics, conditional logics, and intuitionistic modal logics. As a witness to the generality of the methodology, we yet again apply it to a wide range of logics and this time we prove the~ULIP. First, we need to introduce these logics.

Set $\calL_{\triangleright}\ce\{\wedge, \vee, \to, \bot, \triangleright\}$ as the language of conditional logics. For non-normal modal and conditional logics, the property U(L)IP is less explored than in the setting of normal logics. Recall that the basic non-normal modal logic $\Elogic$ (resp.~the conditional logic $\CElogic$) is the smallest set of $\calL_{\Box}$-formulas (resp.~$\calL_{\triangleright}$-formulas) containing classical tautologies, closed under~{\MP}~and~$\Erule$ (resp.~$\CErule$). Define the other non-normal modal and conditional logics as in \cref{figModalHilbert}, by adding the axioms in \cref{figaxioms}. Also, recall the well-known intuitionistic non-normal modal logic, namely the intuitionistic monotone modal logic $\iM$, which is axiomatized over $ \IPC$ by the rule $\Erule$ and the axiom $\Box (\phi \wedge \psi) \impl \Box\phi \wedge \Box \psi$.

\begin{figure}[!t]
\centering{
\small{
\fbox{\parbox{0.97\textwidth}{
\vspace{0.5em}
\ \ \textbf{Modal axioms:}
\vspace{-0.5em}
\begin{align*}
    \Maxiom & \ \ \Box (\phi \wedge \psi) \impl \Box \phi \wedge \Box \psi && & \Caxiom & \ \ \Box \phi \wedge \Box \psi \impl \Box (\phi \wedge \psi) && & \Naxiom & \ \ \Box \top
\end{align*}

\ \ \textbf{Conditional axioms:}
\vspace{-0.5em}
\begin{align*}
    \CM & \ \ (\phi \triangleright \psi \wedge \theta) \impl (\phi \triangleright \psi) \wedge (\phi \triangleright \theta) && & \CN & \ \ \phi \triangleright \top && & \IDaxiom & \ \  \phi \triangleright \phi\\
     \CC & \ \ (\phi \triangleright \psi) \wedge (\phi \triangleright \theta) \impl (\phi \triangleright \psi \wedge \theta) && & \CEM & \ \ (\phi \triangleright \psi) \vee (\phi \triangleright \neg \psi) && & 
\end{align*}

\ \ \textbf{Modal and conditional rules:}
\[
        \vlinf{\Erule}{}{\Box \phi \leftrightarrow \Box \psi}{\phi \leftrightarrow \psi} \qquad \qquad 
        \vliinf{\CErule}{}{\phi_0 \triangleright \psi_0 \impl \phi_1 \triangleright \psi_1}{\phi_0 \leftrightarrow \phi_1} {\psi_0 \leftrightarrow \psi_1} 
\]

}
}}
}
\caption{Modal and conditional axioms and rules.}
\label{figaxioms}
\end{figure}

\begin{figure}[!t]
\centering{
\small{
\fbox{\parbox{0.97\textwidth}{
\vspace{0.5em}
\ \ \textbf{Non-normal modal logics:}
\vspace{-0.5em}
\begin{align*}
    \ENlogic &\ce  \Elogic + \Naxiom   && & \Mlogic &\ce  \Elogic+ \Maxiom && & \MNlogic &\ce  \Mlogic+ \Naxiom && & \MClogic &\ce  \Mlogic + \Caxiom\\
    \K &\ce  \MClogic + \Naxiom   && & \EClogic &\ce  \Elogic+ \Caxiom && & \ECNlogic &\ce  \EClogic+ \Naxiom && & 
\end{align*}

\ \ \textbf{Conditional logics:}
\vspace{-0.5em}
\begin{align*}
    \CENlogic &\ce  \CElogic + \CN   & \CMlogic &\ce  \CElogic+ \CM  & \CMNlogic &\ce  \CMlogic+ \CN  & \CMClogic &\ce  \CMlogic + \CC\\
    \CKlogic &\ce  \CMClogic + \CN    & \CEClogic &\ce  \CElogic+ \CC  & \CECNlogic &\ce  \CEClogic+ \CN  & \CKIDlogic &\ce  \CKlogic+ \IDaxiom
\end{align*}
\vspace{-1.8em}
\[
\qquad \qquad \qquad  \CKCEMlogic\ce  \CKlogic+ \CEM \qquad \qquad \CKCEMIDlogic\ce  \CKCEMlogic+ \IDaxiom
\]
}
}}
}
\caption{Non-normal modal and conditional logics.}
\label{figModalHilbert}
\end{figure}

To prove the~ULIP for logics, we need to extend this concept to sequent calculi, similar to the case of the~UIP. If a sequent calculus has this form of the~ULIP, then its corresponding logic has the~ULIP as well. The summary of the results is provided in the following theorem. \looseness=-1
\begin{theorem}[Positive results, \cite{AkbarTabatabaiIJ22nonnormal,AkbarTabatabaiIJ22int}]
      The following  have the~ULIP (hence the~UIP~and~LIP):\looseness=-1
\begin{itemize}
    \item Non-normal modal logics: $\Elogic$, $\Mlogic$, $\ENlogic$, $\MNlogic$, $\MClogic$, and $\K$.
    \item Conditional logics: $\CElogic$, $\CMlogic$, $\CENlogic$, $\CMNlogic$, $\CMClogic$, $\CKlogic$, and $\CKIDlogic$.
    \item 
The intuitionistic monotone modal logic $\iM$.
\end{itemize}
       \end{theorem}

\begin{theorem}[Negative results,  \cite{AkbarTabatabaiIJ22nonnormal}]\leavevmode
\begin{itemize}
\item 
    The logics $\CKCEMlogic$ and $\CKCEMIDlogic$  do \emph{not} have the~ULIP (but they have the~UIP). Hence, they cannot have a fully terminating semi-analytic calculus.
\item 
Non-normal modal logics $\EClogic$ and $\ECNlogic$ and conditional logics $\CEClogic$ and $\CECNlogic$ do \emph{not} have the~CIP (and hence no U(L)IP). Therefore, they cannot have a semi-analytic calculus.
\end{itemize}    
\end{theorem}

\section{Interpolation Using Generalizations of Sequent Calculi}
\label{sect:gener}

For many logics, no cut-free sequent calculus is known. While certain types of restricted cuts do not impede the Maehara method (see \cref{sect:role_of_cut}), it is often easier to find  a cut-free \emph{hypersequent}, \emph{nested sequent}, or \emph{labelled sequent} calculus for a  logic than a (semi-)analytic sequent calculus. In addition, the former are bettter suited to demonstrate the LIP. In this section we show how to adapt the Maehara method to these generalizations of sequents. \looseness=-1

 We aim to provide a general manual on how to construct interpolants to further explore interpolation properties for other logics of the reader's interest besides the logics treated in this section. We do this by presenting a reverse engineering strategy. What we mean by this is that we will analyze the form of the rules, which will guide us in coming up with the correct constructions of interpolants.

Since hypersequents and nested sequents can usually be translated into the labelled sequent format and there are powerful general methods of constructing cut-free labelled sequent calculi based on Kripke semantics of the logic~\cite{Negri05,NegrivP11}, we first focus on labelled sequents in \cref{sect:multi,sect:lab_interp} for modal logics over~$\CPC$. In \cref{sec:hyper}, we then briefly explain how to adapt this method to  hypersequents and nested sequents.

The method described in this section is sometimes criticized for the explicit use of semantic elements. In this respect, we would like to cite the alternative formulation of the method in~\cite{Lyon_etall20} that avoids such explicit use.

\subsection{Labelled Sequents}
\label{sect:multi}

A \emph{labelled sequent} (over a language~$\calL$) is a sequent (over~$\calL$) enriched with labels: every formula~$\phi$ in the sequent is supplied with some label~$i$, resulting in a labelled formula~$\labe{i}{\phi}$. Labels~$i$ can be taken from any  set. W.l.o.g.,~we will use integer labels from~$\posint$.

The presented method of proving interpolation was originally developed by Fitting and Kuznets for nested sequents~\cite{FittingK15} and later adapted to labelled sequents by Kuznets in~\cite{Kuznets16JELIA}. Here we largely follow~\cite{Kuznets18}, where additional details can be found. As in \cref{sec:basics}, we will present general conditions guaranteeing that interpolation holds, but based on Kripke frame conditions rather than the sequent rule shape.

Like labelled sequents themselves, this method relies on the Kripke semantics. In particular, labels~$i$ are interpreted as Kripke worlds. 
Whereas standard  sequents $\Gamma \seqar \Delta$ consist of two multisets of formulas (antecedent~$\Gamma$ and succedent~$\Delta$), labelled sequents $\relat, \Gamma \seqar \Delta$ are comprised  of two multisets~$\Gamma$~and~$\Delta$ of \emph{labelled formulas} plus  relational part~$\relat$. This~$\relat$ consists of relational atoms of the form $i R j$, where $i$~and~$j$~are labels and $R$~is the Kripke accessibility relation.\footnote{Multiple relations and/or $n$-ary relations may be used according to the Kripke semantics.} $\Lab(S)$~denotes the set of all labels occurring in~$S$. For instance, $\Lab(1 R 2, \ \labe{2}{\phi}, \ \labe{4}{p} \ \seqar \ \labe{4}{\psi})=\{1,2,4\}$. We write $i \in S$ to mean $i \in \Lab(S)$.

\begin{definition}[Labelled semantics]
\label{def:mult_val}
Let $\calM = (W, R, V)$ be a Kripke model. 
For a labelled sequent $S=\relat, \Gamma\seqar\Delta$, a function $\intfun \colon X \to W$ is called an \emph{interpretation of~$S$ into~$\calM$} if{f} $X \supseteq \Lab(S)$, i.e.,~$\intfun$~interprets all labels of~$S$, and $\labint{i} R \labint{j}$ in $\calM$ whenever $i R j \in \relat$.
For such an interpretation, $\calM\intfun \models \relat, \Gamma \seqar \Delta$ if{f} \looseness=-1
\begin{center}
either\quad $\calM, \labint{i} \not\models \phi$ for some  $\labe{i}{\phi} \in \Gamma$\quad or\quad $\calM, \labint{j} \models \psi$ for some $\labe{j}{\psi} \in \Delta$,
\end{center}
in which case we say that $\relat,\Gamma \seqar \Delta$ is true for~$\intfun$. Otherwise, it is false. 
$S$~is valid in a class~$\class$ of models, written $\class \models S$, if{f} $S$~is true for all interpretations~$\intfun$ of~$S$ into models $\calM \in \class$.
\end{definition}

    The notion of validity does not change if  interpretations are defined exactly on the labels~$\Lab(S)$ of a sequent~$S$ rather than on any superset $X \supseteq \Lab(S)$ thereof. 

\begin{remark}
\label{rem:seq_as_lab_seq}
An ordinary sequent $\Gamma 
\seqar \Delta$ can be viewed as a single-label labelled sequent $\labe{1}{\Gamma} \seqar \labe{1}{\Delta}$ with  $\relat=\emptyset$, where $\labe{i}{\Pi}\ce\{\labe{i}{\phi} \mid \phi \in\Pi\}$. This representation is faithful: any mapping~$\seqint$ defined on~$1$ is an interpretation of $\labe{1}{\Gamma} \seqar \labe{1}{\Delta}$ and $\calM\seqint \models \labe{1}{\Gamma} \seqar \labe{1}{\Delta}$ if{f} $\calM, \labint{1} \models \bigwedge\Gamma \impl \bigvee \Delta$.
\lipicsEnd
\end{remark}

Inference rules and proofs for ordinary sequents are defined in  \cref{subsec:seqcal}, so we will use the same notation. Examples of labelled sequent rules can be found in \cref{eq:label_sample_rules,eq:label_modal_rules,eq:label_special_rules} which will be explained in more detail in the next paragraphs. We will typically omit the word \emph{labelled}, while talking about sequents, rules, etc.~to mean the labelled versions and  using \emph{unlabelled} or \emph{ordinary} otherwise. Names of labelled calculi are prefixed with~$\mathbf{L}$ and names of labelled rules are prefixed with~$\mathsf{L}$. 

\begin{figure}[t]
\centering
\small{
\fbox{\parbox{0.98\textwidth}{
{\centering
\vspace{0.5em}
\ \ \ \ \ $\vlinf{\lid^*}{}{\relat, \labe{i}{\phi}, \Gamma \seqar \Delta, \labe{i}{\phi}}{}$
\qquad
$\vlinf{\Ll\lr\wedge}{}{\relat, \labe{i}{\phi\wedge\psi},\Gamma\seqar\Delta}{\relat,\labe{i}{\phi},\labe{i}{\psi},\Gamma\seqar\Delta}$
\qquad
$\vliinf{\Ll\lr\impl}{}{\relat,\labe{i}{\phi \impl \psi},\Gamma\seqar\Delta}{\relat,\Gamma\seqar\Delta,\labe{i}{\phi}}{\relat,\labe{i}{\psi}, \Gamma \seqar \Delta}
$
\vspace{0.5em}
}}}}
\caption{Labelled versions of some propositional sequent rules from \cref{fig:G3K}.}
\label{eq:label_sample_rules}
\end{figure}
\begin{figure}[t]
\centering
\small{
\fbox{\parbox{0.98\textwidth}{
{\centering
\vspace{0.5em}
\qquad \qquad \quad $\vlinf{\Ll\lr\Box}{}{i R j,\relat,  \labe{i}{\Box\phi},\Gamma\seqar\Delta}{i R j,\relat,  \labe{j}{\phi},\labe{i}{\Box\phi},\Gamma\seqar\Delta}$
\qquad\qquad
$\vlinf{\Ll\rr\Box}{\mathsmaller{(j\notin\relat, \Gamma \seqar \Delta)}}{\relat, \Gamma\seqar\Delta, \labe{i}{\Box\phi}}{i R j,\relat,\Gamma\seqar\Delta,  \labe{j}{\phi}}
$
\vspace{0.5em}
}}}}
\caption{Labelled sequent rules for the only modality~$\Box$ of language~$\calLBox$.}
\label{eq:label_modal_rules}
\end{figure}
\begin{figure}[t!]
\centering
\small{
\fbox{\parbox{0.98\textwidth}{
{\centering
\vspace{0.5em}
\qquad $
\vlinf{\LlB}{}{iRj,\relat,  \Gamma\seqar\Delta}{jRi, i R j,\relat,\Gamma\seqar\Delta}$
\qquad
$\vlinf{\Ll{}\Fourseqrule}{}{iRj,jRk,\relat,  \Gamma\seqar\Delta}{iRk,iRj ,jRk,\relat,\Gamma\seqar\Delta}$
\qquad
$\vlinf{\Ll{}\Dseqrule}{\mathsmaller{(j\notin\relat, \Gamma \seqar \Delta)}}{\relat, \Gamma\seqar\Delta}{i R j,\relat,\Gamma\seqar\Delta}$
\vspace{0.5em}
}}}}
\caption{Labelled sequent rules for  frame conditions of symmetry, transitivity, and seriality.}
\label{eq:label_special_rules}
\end{figure}

A labelled sequent calculus typically consists of three groups of rules: propositional rules, modal rules, and special rules. Propositional rules can be obtained by applying the following modification to each of the unlabelled rules from \cref{fig:G3K} (or any other suitable calculus), other than cut.\footnote{We omit the treatment of cut, as labelled sequent calculi are typically cut-free.} For an unlabelled rule, add the same relational part~$\relat$ to the conclusion and all premises, treat multisets~$\Gamma$, $\Delta$,~etc.~as multisets of labelled formulas, and assign the same label to all principal and active formulas. The result for  three of the rules from \cref{fig:G3K} is shown in \cref{eq:label_sample_rules}. The modal rules   require two rules (per modality) and  reflect the Kripke semantics for the respective modality. The two rules  for  language~$\calLBox$ are shown in \cref{eq:label_modal_rules}. 
Special rules are determined on a logic-by-logic basis and can often be generated algorithmically based on the Kripke frame conditions for the logic~\cite{NegrivP11}. For instance, the  rules  added if the frames are symmetric, transitive,  and serial respectively are shown  in \cref{eq:label_special_rules}.\looseness=-1

\begin{definition}
Let modal logic~$\lgc$ be complete w.r.t.~a class~$\class$ of Kripke models. A sequent calculus~$\LSC$ is a \emph{calculus for~$\lgc$ (w.r.t.~$\class$)} 
if{f} $
\LSC \vdash S \Longleftrightarrow \class \models S$ for every  sequent~$S$.
\end{definition}

Labelled versions of the propositional rules from \cref{fig:G3K}, i.e.,~all rules but~$(\Kseqrule)^*$, together with the rules from \cref{eq:label_modal_rules}, form the (cut-free)  calculus~$\bLl\GdrieK$ (over~$\calLBox$) for modal logic~$\K$. Adding to~$\bLl\GdrieK$  any combination of  $(\LlB)$, $(\Ll\Fourseqrule)$, and $(\Ll\Dseqrule)$ from \cref{eq:label_special_rules} yields calculi for the logics complete w.r.t.~the frames of the respective type. Note that some of these logics, e.g,~$\B$, possess no known cut-free \emph{unlabelled} sequent calculus. Cf.~\cite{NegrivP11} for general methods of constructing cut-free labelled sequent calculi.
\begin{proposition}
\label{prop:label_comp}
    If $\LSC$~is a calculus for~$\lgc$, then $\lgc\vdash \phi \impl \psi$ if{f} $\LSC \vdash \labe{i}{\phi} \seqar \labe{i}{\psi}$ for any~$i$.\looseness=-1
\end{proposition}

\subsection{Labelled Interpolation}
\label{sect:lab_interp}

As in the unlabelled case, an interpolant of a theorem $\phi\impl \psi$ is found based on a  proof of a split sequent $\labe{i}{\phi};\; \seqar \;;\labe{i}{\psi}$ (for any label~$i$). Splits are defined as  in \cref{def:split_seq} with relational atoms remaining outside of the split. The left (right) side of the split sequent $\relat, \Gamma; \Gamma' \seqar \Delta; \Delta'$ is $\relat, \Gamma \seqar \Delta$  (resp.~$\relat, \Gamma' \seqar \Delta'$). Split rules are obtained as in the unlabelled case. For the modal rules from \cref{eq:label_modal_rules}, $\Gamma$~and~$\Delta$~are split the same way in the premise and conclusion, whereas both~$\labe{i}{\Box \phi}$~and~$\labe{j}{\phi}$ in the premise have to be on the same side of the split as~$\labe{i}{\Box \phi}$ in the conclusion.

Since labelled formulas  are  interpreted at different Kripke worlds according to a given interpretation~$\intfun$, the same should hold for interpolants, which cannot be mere formulas, but rather boolean combinations\footnote{Conjunctions and disjunctions are often sufficient with additional connectives added if need be.} of labelled formulas: 
\begin{definition}[Multiformula]
 Grammar
$\mho \cce \labe{i}{\phi} \mid (\mho \spdisj \mho) \mid (\mho \spconj \mho)$
where $\phi \in \form_\calL$ and $i\in\posint$ defines \emph{multiformulas} (over a language $\calL$). $i \in \mho$ denotes that label~$i$ occurs in~$\mho$.
\end{definition}
\begin{definition}[Multiformula semantics]
Let $\mho$~be a multiformula and let  $\intfun \colon X \to W$ be an interpretation into  $\calM = (W,R,V)$ with $X \supseteq \Lab(\mho)$. Then $\calM\intfun \models \labe{i}{\phi}$ if{f} $\calM, \labint{i} \models \phi$, while \emph{multiformula conjunction}~$\spconj$ and \emph{disjunction} $\spdisj$ behave classically, i.e., $\calM\intfun \models \mho_1 \spconj \mho_2$ ($\calM\intfun \models \mho_1 \spdisj \mho_2$) if{f} $\calM\intfun \models \mho_j$ for both $j=1,2$ (resp.~for at least one of $j=1,2$). We say that two multiformulas~$\mho$ and $\mho'$ are \emph{logically equivalent}, if{f} for any model~$\calM$ and any interpretation~$\seqint$ of~$\mho$~and~$\mho'$ into~$\calM$, we have $\calM\seqint \models \mho$ if{f} $\calM\seqint \models \mho'$.
\end{definition}

$\Voc{\relat, \Gamma \seqar \Delta}$ and $\Voc{\mho}$ are defined by ignoring~$\relat$, as well as all the labels. 

\begin{definition}[Labelled Craig interpolation]
\label{def:msint}
A  calculus~$\LSC$ over a language~$\calL$ has  the \emph{Labelled  Craig Interpolation Property~(LCIP)} w.r.t.~class~$\class$ of Kripke models if{f}, whenever $\vdash_{\LSC} S$ for a sequent $S= \relat, \Gamma, \Gamma' \seqar \Delta, \Delta'$, there exists a multiformula $\mho$ over $\calL$ such that
\begin{enumerate}
    \item\label{cond:lab_seq_craig}$\Voc{\mho} \subseteq \Voc{\Gamma, \Delta}\cap \Voc{\Gamma', \Delta'}$, i.e., all atoms in $\mho$ occur in both sides of the split;
    \item\label{cond:lab_seq_main} for any interpretation~$\intfun$ of $S$ into a Kripke model $\calM \in \class$:
    \begin{enumerate}[a]
    \item\label{cond:lab_seq_maina} if $\calM\intfun \not\models \mho$, then $\calM\intfun \models \relat, \Gamma \seqar \Delta$ and
    \item\label{cond:lab_seq_mainb} if $\calM\intfun \models \mho$, then $\calM\intfun \models \relat, \Gamma' \seqar \Delta'$;
\end{enumerate}
    \item\label{cond:lab_seq_lab} $\Lab(\mho) \subseteq \Lab(\relat, \Gamma, \Gamma' \seqar \Delta, \Delta')$, i.e., all labels in $\mho$ occur in the sequent.
\end{enumerate}
In this case, we write $\vdash_{\LSC} \relat, \Gamma;\Gamma' \xseqar{\mho} \Delta;\Delta'$ and call $\mho$ the \emph{(labelled Craig) interpolant} for the split sequent $\relat, \Gamma;\Gamma' \seqar \Delta;\Delta'$ in~$\LSC$.
\end{definition}

Here conditions~\ref{cond:lab_seq_craig}--\ref{cond:lab_seq_main} are labelled analogs of the two conditions from \cref{def:Maehara}, with condition~\ref{cond:lab_seq_maina} being  contrapositive relative to the first part of  \cref{def:Maehara}\eqref{cond:seq_main}. Condition~\ref{cond:lab_seq_lab} is labelled-specific (also used for nested and hypersequents). For the~LIP, as in  \cref{def:sequent_Lyndon}:
\begin{definition}[Labelled Lyndon interpolation]
\label{def:lablyndint}
The \emph{Labelled Lyndon Interpolation Property~(LLIP)} is obtained by replacing condition~\ref{cond:lab_seq_craig} in \cref{def:msint} with:
\begin{enumerate}[1'.]
\item\label{cond:lynd_one_label} $\Vocpos{\mho}\subseteq \Vocneg{\Gamma \seqar \Delta}\cap \Vocpos{\Gamma' \seqar \Delta'}$  and $\Vocneg{\mho}\subseteq \Vocpos{\Gamma \seqar \Delta}\cap \Vocneg{\Gamma' \seqar \Delta'}$.
\end{enumerate}
\end{definition}

\begin{theorem}
    Let $\LSC$~be a  calculus for~$\lgc$ (both over~$\calLBox$) w.r.t.~a class~$\class$ of Kripke models. If $\LSC$ has the~LCIP/LLIP, then $\lgc$ has the~CIP/LIP.
\end{theorem}
\begin{proof}[Proof sketch]
    If $\mho$ is a labelled Craig/Lyndon interpolant for $\labe{1}{\phi}; \; \seqar \; ;\labe{1}{\psi}$ in $\LSC$, then one can prove that formula $\multform(\mho)$ is a Craig/Lyndon interpolant for $\phi \impl \psi$ in~$\lgc$, where $\multform(\mho)$ is defined by $\multform(\labe{i}{\phi})\ce \phi$, $\multform(\mho_1 \circledcirc \mho_2) \ce \multform(\mho_1) \circ \multform(\mho_2)$ for $\circ\in\{\wedge,\vee\}$.
\end{proof}

Now we turn to a reverse engineering strategy to see how to define  interpolant transformations for various labelled rules. The form of the rules will guide us on how to come up with a  construction of appropriate interpolant transformations.

For this new presentation we will divide labelled rules into different categories depending on the shape of their split rules. Different categories will require different interpolant transformations. When going through these explanations, the reader might notice many similarities with the interpolant constructions from \cref{fig:interpolation} for unlabelled rules. Indeed, many of the (implicit) principles in \cref{fig:interpolation} transfer to the labelled case with minor modifications. A fundamental difference is that interpolants in the labelled case are multiformulas (contain labels) instead of just formulas, which means that here we need the multiformula variants~$\spconj$~and~$\spdisj$ of~$\wedge$~and~$\vee$ respectively. The reverse engineering strategy leads to the interpolant transformations for propositional labelled rules presented in \cref{fig:lab_interpolation}, which show close similarities with \cref{fig:interpolation} for unlabelled rules.

\begin{figure}[!t]
\centering{
{
\fbox{
\parbox{0.96\textwidth}{
\begin{adjustbox}{scale=0.8}
\parbox{\textwidth}{
\vspace{0.5em}
\textbf{Interpolants for the four splits of   labelled axiom~$(\lid^*)$}
\vspace{-0.5em}
\[
        \vlinf{\itLR(\lid^*)}{}{\relat,\labe{i}{\phi}, \Gamma;\Gamma' \xseqar{\labe{i}{\phi}} \Delta;\Delta',\labe{i}{\phi}}{} 
        \qquad
        \vlinf{\itRR(\lid^*)}{}{\relat,\Gamma;\Gamma',\labe{i}{\phi}\xseqar{\labe{i}{\top}} \Delta ;\Delta',\labe{i}{\phi}}{}         \]
\vspace{-0.9em}
\[
        \vlinf{\itLL(\lid^*)}{}{\relat,\labe{i}{\phi},\Gamma;\Gamma' \ \xseqar{\labe{i}{\bot}}\labe{i}{\phi},\Delta;\Delta'}{}
        \qquad
        \vlinf{\itRL(\lid^*)}{}{ \relat,\Gamma;\Gamma',\labe{i}{\phi}\xseqar{\labe{i}{\neg \phi}}\Delta,\labe{i}{\phi};\Delta'}{}
        \]

\textbf{Interpolant transformations for the two splits of  unary  labelled rule $(\Ll\lr\wedge)$}
\[
\vlinf{L(\Ll\lr\wedge)}{}{\relat,\labe{i}{\phi\wedge\psi},\Gamma ; \Gamma'\xseqar{\mho}\Delta; \Delta'}{\relat,\labe{i}{\phi}, \labe{i}{\psi},\Gamma; \Gamma' \xseqar{\mho}\Delta; \Delta'} 
\qquad
\vlinf{R(\Ll\lr\wedge)}{}{\relat,\Gamma ; \Gamma',\labe{i}{\phi\wedge\psi}\xseqar{\mho}\Delta; \Delta'}{\relat,\Gamma; \Gamma',\labe{i}{\phi}, \labe{i}{\psi} \xseqar{\mho}\Delta; \Delta'}
\]

\textbf{Interpolant transformations for the two splits of binary  labelled rule $(\Ll\lr\impl)$}
\[
\vliinf{L(\Ll\lr\impl)}{}{\relat,\labe{i}{\phi \impl \psi},\Gamma;\Gamma'\xseqar{\mho_1 \spdisj \mho_2}\Delta;\Delta'}{\relat,\Gamma;\Gamma'\xseqar{\mho_1}\labe{i}{\phi},\Delta;\Delta'}{\relat,\labe{i}{\psi}, \Gamma;\Gamma' \xseqar{\mho_2} \Delta;\Delta'} \quad
\vliinf{R(\Ll\lr\impl)}{}{\relat,\Gamma;\Gamma',\labe{i}{\phi \impl \psi}\xseqar{\mho_1 \spconj \mho_2}\Delta;\Delta'}{\relat,\Gamma;\Gamma'\xseqar{\mho_1}\Delta;\Delta',\labe{i}{\phi}}{\relat,\Gamma;\Gamma',\labe{i}{\psi} \xseqar{\mho_2} \Delta;\Delta'}
\]
}
\end{adjustbox}
}
}}}
\caption{Craig/Lyndon interpolant construction  for  sample split  rules obtained from \cref{fig:G3K}. Capital  $L$ ($R$) denotes that the principal formula is on the left (resp.~right) side of the split. } 
\label{fig:lab_interpolation}
\end{figure}

We first discuss three types of interpolant transformations that can be categorized via so-called \textit{local}, \textit{conjunctive}, and \textit{disjunctive} split rules. These categories are defined precisely in \cite{Kuznets18}, but are in essence defined by the three conditions mentioned in the following three lemmas. In particular, a \textit{local split rule} is a single-premise rule that preserves the local logical consequence (if the premise is true for an interpretation, the conclusion must be true for the same interpretation), which leads to the interpolant of the conclusion being the same as the interpolant of the premise.
For a conjunctive (resp.~disjunctive) binary split rule the interpolant of the conclusion is the conjunction (resp.~disjunction) of the interpolants of the premises. \looseness=-1

\begin{lemma}[Local split rules, \cite{Kuznets18}]
\label{lem:local_gen}
Let $\class$ be a class of Kripke models. Assume the following:
\begin{enumerate}[(i)]
    \item $\Lab(\relat^*,\Gamma;\Gamma'\seqar \Delta;\Delta') \subseteq \Lab(\relat,\Pi;\Pi'\seqar \Sigma;\Sigma')$;
    \item \label{eq:local_craig}  $\Voc{\Gamma, \Delta} \subseteq \Voc{\Pi,\Sigma}\quad\text{and}\quad\Voc{\Gamma', \Delta'} \subseteq \Voc{\Pi',\Sigma'}$;
    \item For all $\calM \in \class$ and interpretations $\seqint$ of $\relat,\Pi;\Pi'\seqar \Sigma;\Sigma'$ into $\calM$:
    \begin{align*} 
    \calM\seqint\models\relat^*,\Gamma \seqar \Delta \quad &\Longrightarrow \quad \calM\seqint\models\relat,\Pi\seqar\Sigma, \text{ and }
        \\
        \calM\seqint\models\relat^*,\Gamma' \seqar \Delta' \quad &\Longrightarrow \quad \calM\seqint\models\relat,\Pi'\seqar\Sigma'.
    \end{align*}
\end{enumerate}
Then the following interpolant transformation works for the~LCIP:
\[
\vlinf{}{}{\relat,\Pi;\Pi'\xseqar{\mho}\Sigma;\Sigma'}{\relat^*,\Gamma;\Gamma' \xseqar{\mho} \Delta;\Delta'}\  ,
\]
i.e., if\/ $\mho$ is a labelled Craig interpolant for $\relat^*,\Gamma;\Gamma' \seqar \Delta;\Delta'$, then\/ $\mho$ is also a labelled Craig interpolant for $\relat,\Pi;\Pi'\seqar\Sigma;\Sigma'$.
\end{lemma}

\begin{lemma}[Conjunctive split rules, \cite{Kuznets18}]
\label{lem:conjunct_gen}
Let $\class$ be a class of Kripke models. Assume:
\begin{enumerate}[(i)]
\item\label{eq:conjunct_craig_labels}
$\Lab(\relat_a,\Gamma_a;\Gamma_a'\seqar \Delta_a;\Delta_a') \subseteq \Lab(\relat,\Pi;\Pi'\seqar \Sigma;\Sigma')$ 
\ \text{for both $a \in\{1,2\}$};
\item\label{eq:conjunct_craig}
$\Voc{\Gamma_a, \Delta_a} \subseteq \Voc{\Pi,\Sigma}\quad\text{and}\quad \Voc{\Gamma_a', \Delta_a'} \subseteq \Voc{\Pi',\Sigma'}\  \text{for both $a \in\{1,2\}$}$;
    \item For all $\calM \in \class$ and interpretations $\seqint$ of $\relat,\Pi;\Pi'\seqar \Sigma;\Sigma'$ into $\calM$:
    \begin{align*} 
    \calM\seqint\models\relat_1,\Gamma_1 \seqar \Delta_1  \quad &\Longrightarrow \quad \calM\seqint\models\relat,\Pi\seqar\Sigma, 
        \\
        \calM\seqint\models\relat_2,\Gamma_2 \seqar \Delta_2 \quad &\Longrightarrow \quad \calM\seqint\models\relat,\Pi\seqar\Sigma, \text{ and } \\
   \calM\seqint\models\relat_1,\Gamma_1' \seqar \Delta_1' \ , \ \calM\seqint\models\relat_2,\Gamma_2' \seqar \Delta_2'
    \quad &\Longrightarrow \quad\calM\seqint\models\relat,\Pi'\seqar\Sigma'.
    \end{align*}
\end{enumerate}
Then the following interpolant transformation works for the~LCIP:
\[
{\small
\vliinf{}{}{\relat,\Pi;\Pi'\xseqar{\mho_1 \spconj\mho_2}\Sigma;\Sigma'}{\relat_1,\Gamma_1;\Gamma_1' \xseqar{\mho_1} \Delta_1;\Delta_1'}{\relat_2,\Gamma_2;\Gamma_2' \xseqar{\mho_2} \Delta_2;\Delta_2'}}\ .
\]
\end{lemma}

\begin{lemma}[Disjunctive split rules, \cite{Kuznets18}]
\label{lem:disjunct_gen}
Let $\class$ be a class of Kripke models. Assume   \eqref{eq:conjunct_craig_labels}~and~\eqref{eq:conjunct_craig} from \cref{lem:conjunct_gen}
hold and assume:
\begin{enumerate}[(i)]
\setcounter{enumi}{2}
\item  
For all $\calM \in \class$ and interpretations $\seqint$ of $\relat,\Pi;\Pi'\seqar \Sigma;\Sigma'$ into $\calM$:
    \begin{align*} 
    \calM\seqint\models\relat_1,\Gamma_1 \seqar \Delta_1  \ , \ \calM\seqint\models\relat_2,\Gamma_2 \seqar \Delta_2  
    \quad &\Longrightarrow \quad \calM\seqint\models\relat,\Pi\seqar\Sigma, \\
\calM\seqint\models\relat_1,\Gamma'_1 \seqar \Delta'_1  \quad &\Longrightarrow \quad \calM\seqint\models\relat,\Pi'\seqar\Sigma', \text{ and }
        \\
    \calM\seqint\models\relat_2,\Gamma'_2 \seqar \Delta'_2     \quad &\Longrightarrow \quad \calM\seqint\models\relat,\Pi'\seqar\Sigma'.
    \end{align*}
\end{enumerate}
Then the following interpolant transformation works for the~LCIP:
\[
{\small
\vliinf{}{}{\relat,\Pi;\Pi'\xseqar{\mho_1 \spdisj\mho_2}\Sigma;\Sigma'}{\relat_1,\Gamma_1;\Gamma_1' \xseqar{\mho_1} \Delta_1;\Delta_1'}{\relat_2,\Gamma_2;\Gamma_2' \xseqar{\mho_2} \Delta_2;\Delta_2'}} \ .
\]
\end{lemma}

\normalsize

The  transformations in \cref{lem:local_gen,lem:conjunct_gen,lem:disjunct_gen} also work for the~LLIP if assumptions~\eqref{eq:local_craig} are strengthened to satisfy \cref{def:lablyndint}($1'$), taking polarity into account (for details see~\cite{Kuznets18}).\looseness=-1

Before we provide examples of labelled local, conjunctive, and disjunctive split rules, we note that these concepts can already be observed in the unlabelled case.

\begin{example}[The unlabelled case]
\label{ex:unlabelled-modal-rules}
Recall from \cref{rem:seq_as_lab_seq} that an ordinary sequent can be viewed as a labelled sequent with a single label. Applying the conditions of \cref{lem:local_gen,lem:conjunct_gen,lem:disjunct_gen} to rules in \cref{fig:interpolation}, we see that   $R(\lr\wedge^i)$, $R(\rr\impl)$, and $L(\rr\cntr)$ are local split rules, $L(\lr\vee)$, $L(\rr\wedge)$, and $L(\lr\impl)$ are disjunctive split rules, while $R(\lr\vee)$, $R(\rr\wedge)$, and $R(\lr\impl)$ are conjunctive split rules. Note the symmetry that often occurs for the Left/Right splits of a rule vs.~disjunctive/conjunctive behavior of the interpolant construction. A similar pattern can be seen in the interpolant transformation for the analytic cut rule in the proof of \cref{thm:interpolationSvijf}.

For modal unlabelled rules, consider \cref{fig:modal-rules}. Any split of rule~$(\Tseqrule)$  is local w.r.t.~any class~$\class$ of reflexive models: if $\phi$ is false in a reflexive world, then so is $\Box \phi$. This explains the interpolant preservation for $L(\Tseqrule)$ in \cref{fig:modal_Craig_transformations}.
By contrast,  the other rules in \cref{fig:modal-rules} are not local, hence, require non-trivial transformations in \cref{fig:modal_Craig_transformations}. \looseness=-1 \lipicsEnd
\end{example}

Let us now consider the labelled analog of \cref{ex:unlabelled-modal-rules}.

\begin{example}[The labelled case]
For the propositional rules in \cref{fig:lab_interpolation}, we can identify local, conjunctive, and disjunctive split rules  according to essentially the same argument as in \cref{ex:unlabelled-modal-rules}. For modal labelled rules, the argument is different from the unlabelled case. The standard split version of~$(\Ll\lr\Box)$ from \cref{eq:label_modal_rules} is local (w.r.t.~any class) because $\calM, \labint{j} \not\models \phi$ and $\labint{i}R\labint{j}$ imply $\calM, \labint{i} \not\models \Box \phi$. Consequently, the following interpolant transformations work for the~LCIP/LLIP w.r.t.~any class of models:
\[
{\small
\vlinf{L(\Ll\lr\Box)}{}{i R j,\relat,  \labe{i}{\Box\phi},\Gamma;\Gamma' \xseqar{\mho}\Delta;\Delta'}{i R j,\relat,  \labe{j}{\phi},\labe{i}{\Box\phi},\Gamma;\Gamma'\xseqar{\mho}\Delta;\Delta'}
\qquad
\vlinf{R(\Ll\lr\Box)}{}{i R j,\relat,  \Gamma;\Gamma', \labe{i}{\Box\phi}\xseqar{\mho}\Delta;\Delta'}{i R j,\relat,  \Gamma;\Gamma',\labe{i}{\Box\phi},\labe{j}{\phi}\xseqar{\mho}\Delta;\Delta'}
}
\]
However, the left and right split versions of $(\Ll\rr \Box)$ cannot be identified as  local, conjunctive, or disjunctive and need different treatments as discussed in \cref{Ex: box diamond}. \lipicsEnd
\end{example}

In addition, many of the unary special labelled rules are local, hence, require no change to the interpolant. A large class of special rules is obtained by so-called Horn clauses. Restricting a class of frames by a Horn clause of the form  $i_1 R j_1 \wedge \dots \wedge i_m R j_m \impl k R l$  corresponds to adding to the labelled calculus the rule
\begin{equation}
\label{eq:Horn_rule}
{\small
\vlinf{}{\mathsmaller{(k,l \in i_1Rj_1, \dots, i_mRj_m, \relat, \Gamma \seqar \Delta)}}{i_1Rj_1, \dots, i_mRj_m, \relat,\Gamma \seqar \Delta}{kRl, i_1Rj_1, \dots, i_mRj_m, \relat, \Gamma \seqar \Delta}
}
\end{equation}
along with the
rules obtained from it by the \emph{closure condition}, i.e., by  contracting identical relational atoms $i_aRj_a$ and $i_bRj_b$  from the conclusion in both premise and conclusion for those instances of the rule that contain such identical atoms~\cite[Th.~11.27, Cor.~11.29]{NegrivP11}. Note that in the absence of a principal formula, such rules only have one corresponding split rule (instead of a left and right split version).

\begin{example}[Horn frame conditions, \cite{Kuznets16JELIA}]
\label{lem:horn}
For any split of~\eqref{eq:Horn_rule} or  its contractions, the following interpolant transformation works for the~LCIP/LLIP w.r.t.~any class of models satisfying the condition $i_1 R j_1 \wedge \dots \wedge i_m R j_m \impl k R l$:
\[
{\small
 \vlinf{}{}{i_1Rj_1, \dots, i_mRj_m, \relat,\Gamma;\Gamma' \xseqar{\mho} \Delta;\Delta'}{kRl, i_1Rj_1, \dots, i_mRj_m, \relat, \Gamma;\Gamma' \xseqar{\mho} \Delta;\Delta'}
 }
\]
Consequently, the split version of rule $(\LlB)$ for any class of symmetric frames and of  rule~$(\Ll\Fourseqrule)$ for any class of transitive frames from \cref{eq:label_special_rules} are local and, therefore, the above interpolant transformation works for the~LCIP/LLIP. \lipicsEnd
\end{example}

Two  rules from \cref{eq:label_modal_rules,eq:label_special_rules} do not fit any of the patterns described so far: the  premise  of both $(\Ll\rr\Box)$ and $(\Ll\Dseqrule)$ has a fresh label not occurring in the conclusion. Such rules require a more complex transformation to remove the extra label from premise interpolants. For that, we introduce additional notation for modifying interpretations.

\begin{definition}
Let $\calM=(W,R,V)$ be a Kripke model and $\seqint\colon X \to W$ map labels from \mbox{$X \subseteq \posint$} to worlds of $\calM$. For  $v \in W$ and $j \in \posint$, we define
$
\seqint_{j}^{v} \ce \bigl(\seqint \setminus (\{j\} \times W)\bigr)\sqcup \{(j,v)\}
$
where $\sqcup$ is used to emphasize that this is the union of disjoint sets. 
\end{definition}

The domain of $\seqint_{j}^{v}$ is $X \cup \{ j\}$. Note that the only possible difference between $\seqint_{j}^{v}$ and  $\seqint$ is that the former maps $j$  to $v$. Note also that if $j\notin\Lab(S)$ for a labelled sequent $S$, then $\seqint$~is an interpretation for $S$ if{f} $\seqint_{j}^{v}$ is and, if both are,  $\calM\seqint\models S$ if{f} $\calM\seqint_{j}^{v}\models S$ for any $v \in W$. The same applies to multiformulas.

\begin{definition}
Let $\calM=(W,R,V)$ be a Kripke model and $\seqint\colon X \to W$ be an interpretation of a sequent $S$ into $\calM$. For $i \in X$ and  $j \ne i$, we write
\begin{align*}
\calM\seqint_{j}^{i R \forall} \models S
& \qquad \text{ whenever } \qquad
\forall v \in W \text{ such that } \labint{i}Rv, 
\text{ we have }\calM\seqint_{j}^{v} \models S
\\
\calM\seqint_{j}^{i R \exists} \models S
& \qquad \text{ whenever } \qquad
\exists v \in W \text{ such that } \labint{i}Rv \text{ and } \calM\seqint_{j}^{v} \models S.
\end{align*}
\end{definition}

\begin{definition}
A label $j$ is  \emph{$\spconj\spdisj$-separated} (resp.~\emph{$\spdisj\spconj$-separated}) in a multiformula $\mho$  if{f} for some formulas $B_1,\dots,B_m$ and multiformulas $\mho_1,\dots,\mho_m$ with  $j \notin \Lab(\mho_k)$ for any $1 \leq k \leq m$, we have $\mho = \bigspconj_{k=1}^m (\labe{j}{B_k} \spdisj \mho_k)$ (resp.~$\mho = \bigspdisj_{k=1}^m (\labe{j}{B_k} \spconj \mho_k)$).
\end{definition}
\begin{proposition}
\label{prop:label_norm_form}
Let $j \in \posint$. For any multiformula $\mho$, there is a logically equivalent multiformula $\mho'$ such that $j$~is $\spconj\spdisj$-separated  (resp.~$\spdisj\spconj$-separated) in~$\mho'$. 
\end{proposition}
\begin{proof}
$\mho'$ can be constructed by rewritings similar to the ones used for CNF (resp.~DNF), in view of logical equivalences of $\labe{j}{\phi}\spconj\labe{j}{\psi}$  to $\labe{j}{\phi \wedge\psi}$ and of $\labe{j}{\phi}\spdisj\labe{j}{\psi}$  to $\labe{j}{\phi \vee\psi}$. If $\labe{j}{B_k}$ is absent for some $k$, either $\labe{j}{\top}\spconj$ or $\labe{j}{\bot}\spdisj$ is added, whichever is logically equivalent.
\end{proof}

\Cref{prop:label_norm_form} does not guarantee the uniqueness of $\mho'$, but given $\mho$ we know there is at least one logically equivalent multiformula $\mho'$ with $j$ $\spconj\spdisj$-separated  (resp.~$\spdisj\spconj$-separated) in it. From now on, given a multiformula $\mho$, we allow ourselves to pick one such $\mho'$ and denote it by $\mho_{\spconj j \spdisj}$ (resp.~$\mho_{\spdisj j \spconj}$). For a multiformula $\mho$, we further write $\mho^{j 
\mapsto i\Box}$ (resp.~$\mho^{j 
\mapsto i\Diamond}$) to mean the result of replacing each occurrence of $\labe{j}{}$ in $\mho$ with $\labe{i}{\Box}$ (resp.~with~$\labe{i}{\Diamond}$). 

\begin{example}
\label{ex:mod_trans}
Suppose $\mho_{\spconj j \spdisj} = \bigspconj_{k=1}^m (\labe{j}{B_k} \spdisj \mho_k)$ (resp.~$\mho_{\spdisj j \spconj} = \bigspdisj_{k=1}^m (\labe{j}{B_k} \spconj \mho_k)$) with $j \notin \Lab(\mho_k)$ for any $1 \leq k \leq m$. Then, 
\begin{align*}
\mho_{\spconj j\spdisj}^{j 
\mapsto i\Box} = \bigspconj_{k=1}^m (\labe{i}{
\Box B_k} \spdisj \mho_k)
\qquad
\big(\text{respectively}
\qquad
\mho_{\spdisj j\spconj}^{j 
\mapsto i\Diamond} = \bigspdisj_{k=1}^m (\labe{i}{
\Diamond B_k} \spconj \mho_k) \big).
\end{align*}
Clearly,  $j \notin\Lab(\mho_{\spconj j\spdisj}^{j 
\mapsto i\Box})$ (resp.~$j \notin\Lab(\mho_{\spdisj j\spconj}^{j 
\mapsto i\Diamond})$). \lipicsEnd
\end{example}

The general transformation applicable to rules $(\Ll\rr\Box)$ and $(\Ll\Dseqrule)$ additionally exploits the fact that the fresh label $j$ is only connected to one existing label $i$ by one relational atom $iRj$.
\begin{lemma}[$\Box$-like rules, \cite{Kuznets18}]
\label{lem:box_gen}
Let $\class$ be a class of Kripke models. Assume:
\begin{enumerate}[(i)]
\item\label{eq:box_craig_labels_zero}
$\Voc{\Gamma, \Delta} \subseteq \Voc{\Pi,\Sigma}\quad\text{and}\quad\Voc{\Gamma', \Delta'} \subseteq \Voc{\Pi',\Sigma'}$,
\item\label{eq:box_craig_labels_one}
$j \notin \Lab(\relat,\Pi;\Pi'\seqar \Sigma;\Sigma')
\quad\text{and}\quad 
\{i\} \cup \Lab(\relat^*) \subseteq \Lab(\relat,\Pi;\Pi'\seqar \Sigma;\Sigma')$, 
\item\label{eq:box_craig_labels_two}
$\Lab(\Gamma;\Gamma'\seqar \Delta;\Delta') \subseteq \{j\} \sqcup \Lab(\relat,\Pi;\Pi'\seqar \Sigma;\Sigma')$,
\item 
For all $\calM \in \class$ and interpretations $\seqint$ of $\relat,\Pi;\Pi'\seqar \Sigma;\Sigma'$ into $\calM$:
    \begin{align*} 
\calM\seqint^{iR\exists}_j\models iRj,\relat^*,\Gamma \seqar \Delta \quad &\Longrightarrow \quad \calM\seqint\models\relat,\Pi\seqar\Sigma, \ \text{and} \\
\calM\seqint^{iR\forall}_j\models iRj,\relat^*,\Gamma' \seqar \Delta' \quad & \Longrightarrow \quad \calM\seqint\models\relat,\Pi'\seqar\Sigma'.
\end{align*}
\end{enumerate}
Then the following interpolant transformation works for the~LCIP:
\[
\vlinf{}{}{\relat,\Pi;\Pi'\xseqar{\mho_{\spconj j \spdisj}^{j \mapsto \Box i}}\Sigma;\Sigma'}{iRj,\relat^*,\Gamma;\Gamma' \xseqar{\mho} \Delta;\Delta'} \ .
\]  
\end{lemma}
\begin{lemma}[$\Diamond$-like rules, \cite{Kuznets18}]
\label{lem:diamond_gen}
Let $\class$ be a class of Kripke models. Assume \eqref{eq:box_craig_labels_zero},\eqref{eq:box_craig_labels_one}, and \eqref{eq:box_craig_labels_two} from \cref{lem:box_gen} hold. Moreover, assume:
\begin{enumerate}[(i)]
\setcounter{enumi}{3}
    \item For all $\calM \in \class$ and interpretations $\seqint$ of $\relat,\Pi;\Pi'\seqar \Sigma;\Sigma'$ into $\calM$:
    \begin{align*}
\calM\seqint^{iR\forall}_j\models iRj,\relat^*,\Gamma \seqar \Delta \quad & \Longrightarrow \quad \calM\seqint\models\relat,\Pi\seqar\Sigma, \ \text{and} \\
\calM\seqint^{iR\exists}_j\models iRj,\relat^*,\Gamma' \seqar \Delta' \quad & \Longrightarrow \quad \calM\seqint\models\relat,\Pi'\seqar\Sigma'.
    \end{align*}
\end{enumerate}
Then the following interpolant transformation works for the~LCIP:
\[
\vlinf{}{}{\relat,\Pi;\Pi'\xseqar{\mho_{\spdisj j \spconj}^{j \mapsto \Diamond i}}\Sigma;\Sigma'}{iRj,\relat^*,\Gamma;\Gamma' \xseqar{\mho} \Delta;\Delta'} \ .
\] 
\end{lemma}

The  transformations in \cref{lem:box_gen,lem:diamond_gen} also work for the~LLIP if assumption~\eqref{eq:box_craig_labels_zero} is strengthened to satisfy \cref{def:lablyndint}($1'$), taking polarity into account (for details see~\cite{Kuznets18}).\looseness=-1
\begin{proof}
We treat \cref{lem:diamond_gen} (\cref{lem:box_gen} is dual). Here we only prove property~\ref{cond:lab_seq_main} from \cref{def:msint}. Let $\seqint$ be an interpretation for $\relat,\Pi;\Pi'\seqar \Sigma;\Sigma'$. Suppose that $\calM\seqint\not\models\mho_{\spdisj j \spconj}^{j \mapsto \Diamond i}$,  where multiformulas $\mho_{\spdisj j\spconj}$ and $\mho_{\spdisj j \spconj}^{j \mapsto \Diamond i}$ are as in \cref{ex:mod_trans}.  Then, for each $1\leq k \leq m$, either $\calM, \labint{i} \not \models \Diamond B_k$ or $\calM\seqint \not\models \mho_k$. Thus, for all $v$ such that $\labint{i}Rv$ either $\calM, v \not \models B_k$ or $\calM\seqint_j^{v} \not\models \mho_k$
(since $j \notin \Lab(\mho_k)$). Consequently, $\calM\seqint_j^{v}\not \models\mho_{\spdisj j\spconj}$, i.e., $\calM\seqint_j^{v}\not \models\mho$ for all $v$ such that $\labint{i}Rv$. Given that $\mho$ is an interpolant of the premise,  $\calM\seqint^{iR\forall}_j\models iRj,\relat^*,\Gamma \seqar \Delta$, which yields $\calM\seqint\models\relat,\Pi\seqar\Sigma$ by the first assumed validity.

Suppose now that that $\calM\seqint\models\mho_{\spdisj j \spconj}^{j \mapsto \Diamond i}$.  Then, there is some  $1 \leq k \leq m$ such that  $\calM, \labint{i} \models \Diamond B_k$ and $\calM\seqint \models \mho_k$. From the former, we find a $v$ such that $\labint{i}Rv$, and $\calM, v \models B_k$, and $\calM\seqint_j^{v} \models \mho_k$.
Consequently, $\calM\seqint_j^{v} \models\mho_{\spdisj j\spconj}$, i.e., $\calM\seqint_j^{v} \models\mho$ for this $v$. Given that $\mho$ is an interpolant of the premise, we get $\calM\seqint^{iR\exists}_j\models iRj,\relat^*,\Gamma' \seqar \Delta'$, which yields the requisite $\calM\seqint\models\relat,\Pi'\seqar\Sigma'$ by the second assumed validity.
\end{proof}
\begin{example}[$\Box$- and $\Diamond$-like rules]\label{Ex: box diamond}
The split version of $(\Ll\rr\Box)$ is $\Box$-like (resp.~$\Diamond$-like) w.r.t.~any class of models if $\Box \phi$ is on the right (resp.~left) side of the split. Consequently, the following interpolant transformations work for the~LCIP/LLIP w.r.t.~any class of models: 
\[
\vlinf{L(\Ll\rr\Box)}{}{\relat, \Gamma;\Gamma'\xseqar{\mho_{\spdisj j \spconj}^{j \mapsto \Diamond i}}\labe{i}{\Box\phi},\Delta;\Delta'}{i R j,\relat,\Gamma;\Gamma'\xseqar{\mho}\labe{j}{\phi},\Delta;\Delta'}
\qquad
\vlinf{R(\Ll\rr\Box)}{}{\relat, \Gamma;\Gamma'\xseqar{\mho_{\spconj j \spdisj}^{j \mapsto \Box i}}\Delta;\Delta',\labe{i}{\Box\phi}}{i R j,\relat,\Gamma;\Gamma'\xseqar{\mho}\Delta;\Delta',\labe{j}{\phi}}
\]
where $i$ does  and $j$ does not occur in the conclusion.
Rule~$(\Ll\Dseqrule)$ is interesting because there is no separation into the right and left split variants. Its only split version  is both $\Box$-like and $\Diamond$-like as long as the frame is serial. (The seriality is needed to avoid the vacuous fulfillment of the $\calM\seqint^{iR\forall}_j\models$ parts in \cref{lem:box_gen,lem:diamond_gen}.) Thus, for $i$ occurring  and $j$ not occurring in the conclusion, w.r.t. any class $\class$ of serial frames, both transformations are valid:
\[
\vlinf{(\Ll{}D)}{}{\relat, \Gamma;\Gamma'\xseqar{\mho_{\spdisj j \spconj}^{j \mapsto \Diamond i}}\Delta;\Delta'}{i R j,\relat,\Gamma;\Gamma'\xseqar{\mho}\Delta;\Delta'}
\qquad
\vlinf{(\Ll{}D)}{}{\relat, \Gamma;\Gamma'\xseqar{\mho_{\spconj j \spdisj}^{j \mapsto \Box i}}\Delta;\Delta'}{i R j,\relat,\Gamma;\Gamma'\xseqar{\mho}\Delta;\Delta'} \ . 
\]
Note that the $\Box \theta$ and $\Diamond \theta$ transformations in \cref{fig:modal_Craig_transformations} are the unlabelled manifestations of $\Box$- and $\Diamond$-like rules.\lipicsEnd
\end{example}

\begin{example}
\label{ex:labelled-interpolant}
Let $\phi = \Box\Box(p \wedge \Box q)$ and $\psi=\Box(q \vee r)$. In \cref{fig:lab_ex}, we find an interpolant of $\phi \impl \psi$ using the split labelled calculus for logic $\DB$ (the logic of symmetric serial frames), which has no cut-free sequent calculus.
    \begin{figure}[t]
    \centering
    {$\footnotesize{
\vlderivation{
\vlin{R(\Ll\rr\Box)}{}{\labe{1}{\Box\Box(p\wedge \Box q)}  ; \ \xseqar{\labe{1}{\Box q
    } \spdisj \labe{1}{\bot}} \ ; \labe{1}{\Box (q \vee r)}}
{
    \vlin{L(\Ll\lr\Box)}{}{
    1R2, \labe{1}{\Box\Box(p\wedge \Box q)} ; \ \xseqar{\labe{2}{\Diamond \top}\spconj \labe{2}{q}\equiv \labe{2}{q
    } \spdisj \labe{1}{\bot}} \ ; \labe{2}{(q \vee r)}
    }{
        \vlin{\Ll\Dseqrule}{}{1R2, \labe{2}{\Box(p\wedge \Box q)}, \labe{1}{\Box\Box(p\wedge \Box q)} ; \ \xseqar{\labe{2}{\Diamond \top}\spconj \labe{2}{q}} \ ; \labe{2}{(q \vee r)}}{
            \vlin{L(\Ll\lr\Box)}{}{2R3, 1R2, \labe{2}{\Box(p\wedge \Box q)}, \labe{1}{\Box\Box(p\wedge \Box q)} ; \ \xseqar{\labe{2}{q}\equiv \labe{3}{\top}\spconj \labe{2}{q}} \ ; \labe{2}{(q \vee r)}}{
                \vlin{L(\Ll\lr\wedge)}{}{2R3, 1R2, \labe{3}{(p\wedge \Box q)}, \labe{2}{\Box(p\wedge \Box q)}, \labe{1}{\Box\Box(p\wedge \Box q)} ; \ \xseqar{\labe{2}{q}} \ ; \labe{2}{(q \vee r)}}{
                    \vlin{\LlB}{}{2R3, 1R2, \labe{3}{p}, \labe{3}{\Box q}, \labe{2}{\Box(p\wedge \Box q)}, \labe{1}{\Box\Box(p\wedge \Box q)} ; \ \xseqar{\labe{2}{q}} \ ; \labe{2}{(q \vee r)}}{
                        \vlin{L(\Ll\lr\Box)}{}{3R2, 2R3, 1R2, \labe{3}{p}, \labe{3}{\Box q}, \labe{2}{\Box(p\wedge \Box q)}, \labe{1}{\Box\Box(p\wedge \Box q)} ; \ \xseqar{\labe{2}{q}} \ ; \labe{2}{(q \vee r)}}{
                            \vlin{R(\Ll\rr\vee)}{}{3R2, 2R3, 1R2, \labe{2}{q},\labe{3}{p}, \labe{3}{\Box q}, \labe{2}{\Box(p\wedge \Box q)}, \labe{1}{\Box\Box(p\wedge \Box q)} ; \ \xseqar{\labe{2}{q}} \ ; \labe{2}{(q \vee r)}}{
                            \vlin{\itLR(\lid^*)}{}{3R2, 2R3, 1R2, \labe{2}{q},\labe{3}{p}, \labe{3}{\Box q}, \labe{2}{\Box(p\wedge \Box q)}, \labe{1}{\Box\Box(p\wedge \Box q)} ; \ \xseqar{\labe{2}{q}} \ ; \labe{2}{q},\labe{2}{r}}
                            {
                            \vlhy{}
                            }
                            }
                        }
                    }
                }
            }
            }   
        }
    }
}
}
$}
\caption{Labelled interpolant computation in \cref{ex:labelled-interpolant}.}
\label{fig:lab_ex}
\end{figure}
Thus,  $\multform(\labe{1}{\Box q} \spdisj \labe{1}{\bot})=\Box q \vee \bot$, or equivalently~$\Box q$, is a Lyndon interpolant of $\phi \impl \psi=\Box\Box(p \wedge \Box q)\impl\Box(q \vee r)$ in logic~$\DB$. Here rule~$(\Ll{}D)$ was treated as a $\Diamond$-like, with label~$3$ being $\spdisj\spconj$-separated in $\labe{3}{\top}\spconj\labe{2}{q}\equiv\labe{2}{q}$. Similarly, label~$2$ is $\spconj\spdisj$-separated in $\labe{2}{q} \spdisj \labe{1}{\bot}\equiv\labe{2}{\Diamond \top} \spconj \labe{2}{q}$.  (For this logic, $\Diamond \top$ may be dropped due to seriality. In, e.g.,~$\K$, one would use $\labe{2}{(\Diamond \top \wedge q)} \spdisj \labe{1}{\bot}$ instead.)\looseness=-1 \lipicsEnd
\end{example}

In the next section we see how these techniques lead to powerful interpolation results.

\subsection{Interpolation Results for Labelled, Hyper-, and Nested Sequents}
\label{sec:hyper}
The general approach explained in the previous section of defining interpolants by looking at the shape of labelled split rules has two advantages. First, it leads to interpolation results (including Lyndon interpolation) for a range of logics. Second, the techniques can be adjusted to other types of sequents such as nested sequents and hypersequents. In this section we discuss results in both directions and point to relevant literature for more details.

Composing the algorithm for generating cut-free labelled sequent calculi rules \cite{NegrivP11} with the algorithm discussed in \cref{sect:lab_interp} for constructing interpolant transformations can yield general theorems as follows:
\begin{theorem}[\cite{Kuznets18}]
The modal logic $\lgc_\class$ of any class~$\class$ of frames defined by a set of frame conditions in the Horn-clause form $i_1 R j_1 \wedge \dots \wedge i_m R j_m \impl k R l$ has the~LIP.
\end{theorem}
\begin{corollary}
\label{cor:modal_logics}
Of the frame conditions listed in~\cite[Sect.~8]{Garson23}, any combination of seriality, reflexivity, transitivity, symmetry, euclideanity, and shift reflexivity produces a logic with the~LIP. This includes all 15 modal logics of the so-called modal cube~\cite[Sect.~8]{Garson23}, among them $\K$, $\T$, $\D$, $\Kvier$, $\Svier$, and $\Svijf$ from \cref{fig:modal-logics}.
\end{corollary}

The range of Kripke-frame classes guaranteeing the~LIP can  be expanded further, in particular to deal with the convergence restriction~\cite[Sect.~8]{Garson23} and $hijk$-convergence~\cite[Sect.~9]{Garson23}, yielding logics often named after Scott--Lemon or Geach.
 For details we refer to~\cite{Kuznets18}.

With \cref{cor:modal_logics} we can illustrate that labelled sequent calculi can be more suitable to obtain Lyndon interpolation than ordinary sequent calculi with some form of cut. For example, in \cref{thm:interpolationSvijf} we used Maehara's method to prove Craig interpolation for logic~$\Svijf$ by using the analytic cut rule. However,   analytic cuts  do not preserve polarity, resulting in the inability to obtain  Lyndon interpolation. \Cref{cor:modal_logics}, by contrast, establishes the~LIP for $\Svijf$ using the techniques for labelled sequents.

Now we turn to interpolation results via hypersequents and nested sequents. Nested sequents and hypersequents are closely related to labelled sequents. A labelled sequent $S = \relat,\Gamma\seqar\Delta$ can  be viewed as consisting of ordinary, unlabelled sequents $\{\Gamma_i \seqar \Delta_i\}_{i \in \Lab(S)}$  where multisets $\Gamma_i \ce \{\phi \mid \labe{i}{\phi}\in \Gamma\}$ and $\Delta_i \ce \{\phi \mid \labe{i}{\phi}\in \Delta\}$ contain (unlabelled) formulas. In this view, each sequent $\Gamma_i \seqar \Delta_i$ is a  \emph{sequent component of~$S$} and relational atoms from~$\relat$  describe the relationships among these sequent components. Relational atoms are expressive enough to represent any (Kripke-like) relationship. If a particular relationship is regular enough, it can be encoded by structural connectives, bypassing the use of explicit labels. 

\subparagraph{Hypersequents} 
A \emph{hypersequent} is defined as a finite multiset (alternatively, a set or a list) of sequent components, commonly written as $\Gamma_1 \seqar \Delta_1 \mid \dots \mid \Gamma_n \seqar \Delta_n
$. For example, modal logic $\Svijf$ admits a hypersequent calculus where hypersequents can be seen as labelled sequents in which $\relat$ describes a total relation (e.g., $iRj \in \relat$ for all $1 \leq i,j \leq n$). Although hypersequents predate labelled sequents, having been  introduced in \cite{Minc71,Pottinger83,Avron96}, the labelled view on hypersequents means that any interpolant transformation from \cref{sect:lab_interp} can be used whenever its conditions are satisfied. To use the interpolation algorithm, labels have to be assigned to each sequent component of each hypersequent in a proof. One may need to be careful when assigning labels, depending on whether a rule modifies a sequent component or replaces it with a different component. The latter case calls for using a different label, e.g.,~the proper labelling of Avron's $(\seqar \Box)$ rule from~\cite{Avron96} can be understood as follows:
\[
\vlinf{\tiny \seqar\Box}{}{G \mid \Box \Gamma\seqar \Box \phi}
{G \mid \Box \Gamma\seqar \phi}
\quad\rightsquigarrow\quad
\vlinf{}{}{G \mid \underbrace{\Box \Gamma\seqar \Box \phi}_i}
{G \mid \overbrace{\strut\ \ \ }^i \mid \overbrace{\Box \Gamma\seqar \phi}^j}
\quad\rightsquigarrow\quad
\vlinf{\tiny \Ll\seqar\Box}{j \notin X \sqcup\{i\}}{ \Pi,  \labe{i}{\Box \Gamma}\seqar \Sigma,\labe{i}{\Box \phi}}
{\Pi,  \labe{j}{\Box \Gamma}\seqar \Sigma,\labe{j}{\phi}}
\]
where $X=\Lab(G)$ is the set of labels assigned to the side components from $G$.
Consequently, $j$~is a label to be removed from the premise interpolant, and the two split versions of this rule are $\Box$-like and $\Diamond$-like, rather than local. See~\cite{Kuznets16Kolleg}~and~\cite[Sect.~3]{Kuznets18} for more details and interpolant transformations for other rules.
The proof-theoretic methods to prove the~LIP directly via hypersequents have been applied to~$\Svijf$ and modal logic~\textsf{S4.2} in~\cite{Kuznets16Kolleg,Kuznets18}. 

\subparagraph*{Nested Sequents, or Tree-Hypersequents}
Another commonly used configuration of sequent components is a tree of sequent components~\cite{Bruennler09,Poggiolesi09,LellmannP24}, usually under the names of \emph{nested sequents} or \emph{tree-hypersequents}. Tree structures  can be conveniently represented by labels consisting of sequences of integers, so it is no surprise that prefixed tableaux relate to nested sequents the same way as ordinary tableaux do to ordinary sequents~\cite{Fitting12}. In this case, the implicit relational atoms are of the form $\sigma R \sigma.n$, where $\sigma$~is a sequence of integers and $\sigma.n$~is obtained by appending it with integer~$n$. All the details can be found in~\cite{FittingK15} and \cite[Sect.~5]{Kuznets18}. Here, we only show  an instance of a nested rule whose splits are $\Box$- and $\Diamond$-like rules respectively:\looseness=-1
\[
\vlinf{}{}{\Pi\vlfill{\Gamma\seqar \Delta, \Box \phi}}
{\Pi\vlfill{\Gamma\seqar \Delta, [\phi]}}
\qquad\rightsquigarrow\qquad
\vlinf{}{\sigma.j \notin \relat, \Pi^\ell,  \labe{\sigma}{\Gamma}\seqar \labe{\sigma}{\Delta},\labe{\sigma}{\Box \phi}}{\relat, \Pi^\ell,  \labe{\sigma}{\Gamma}\seqar \labe{\sigma}{\Delta},\labe{\sigma}{\Box \phi}}
{\sigma R\sigma.j,\relat,\Pi^\ell,  \labe{\sigma}{\Gamma}\seqar \labe{\sigma}{\Delta},\labe{\sigma.j}{\phi}}
\]
where $\Pi^\ell$ is the labelled translation of the nested context $\Pi\vlhole$, sequent $\Gamma \seqar \Delta, \Box \phi$ is an ordinary sequent labelled with an integer sequence $\sigma$, and $\relat$ describes the tree structure of the nested tree of $\Pi\vlhole$. Here, label $\sigma.j$ is the one to be removed from the premise interpolant. The proof-theoretic methods to prove the~LIP directly via nested sequents have been applied~\cite{FittingK15,Kuznets18} to all 15 logics of the modal cube from~\cite[Sect.~8]{Garson23}.

\subparagraph{Intermediate Logics}
All formalisms so far in \cref{sect:gener} dealt with $\CPC$-based modal logics. 
The method  can also be applied to intermediate logics, with intuitionistic implication and negation playing the role of modality in modal logics. In most cases, this requires specialized, restricted types of nested sequents and/or hypersequents (see~\cite{KuznetsL18,KuznetsL21} for details). All seven interpolable intermediate logics can be dealt with using one calculus or another. Moreover, the~LIP for the G\"odel--Dummett logic, the logic of linear intuitionistic Kripke frames, was first demonstrated using this method~\cite{KuznetsL18}.

\subparagraph{Uniform interpolation}
Hypersequents and nested sequents can also be used to demonstrate the~UIP/ULIP. This has been done to reprove the UIP for logics~$\K$, $\D$, and~$\T$, and to provide the first proof-theoretic proof of the~ULIP for logic~$\Kvijf$~\cite{vdGiessenJK24,vdGiessenJK23TABLEAUX}. Similar to the labelled sequent definition of Craig interpolation (\Cref{def:msint}), a labelled version of uniform interpolation is incorporating semantic notions, based on bisimulation quantifiers.

\section{Conclusion}
\label{sec:conclusion}

This chapter provides methods for proving interpolation results that simultaneously provide algorithms for constructing interpolants based on various (generalized) sequent calculi. For each algorithm, standard interpolant transformations are presented for common rule types, which cover a wide range of logics. As a result, the chapter can be used as a manual: To prove a particular interpolation property for a logic represented by a given calculus, the rules of the calculus should be checked against the common rule types. If each of the calculus rules is covered by one of these types, the interpolation property directly follows. Otherwise, to prove interpolation, it is sufficient to develop interpolant transformations only for the rules not covered by the common types. The proof-theoretic techniques for sequent calculi are further embedded within the framework of universal proof theory, which systematically relates interpolation properties to the existence of well-behaved proof systems.\looseness=-1 

It is a common misconception that the existence of a cut-free (sequent)  calculus for a logic necessarily implies that the logic enjoys interpolation. It may be impossible to provide an interpolant transformation for a particular rule, even for an analytic one. An example of such non-interpolable hypersequent rule for the G\"odel--Dummett logic can be found in the proof of \cite[Theorem~4.6]{KuznetsL18}. Note that this logic does enjoy interpolation, both Craig and Lyndon. Thus,
more complex calculi may  be necessary despite the existence of a cut-free sequent one.\looseness=-1 

\subparagraph{Pointers to the literature}
For related reading on the discussed material in this chapter we provide the reader with some pointers to the literature: %\iris{\textit{I phrase the following as a summation where I don't use proper sentences. Is that fine?}} \iris{\textit{Are there more items to add?}} 
\begin{itemize}
    \item \textit{Maehara's method:} Maehara's original work \cite{Maehara61}. Textbooks on proof theory with Craig interpolation proofs via sequents for classical/intuitionistic propositional logic and classical/intuitionistic first-order logic~\cite{Takeuti75,TroelstraS00,Schuette77}.  See also the lecture notes in \cite{vdGiessen25} on which \Cref{sec:basics} is based. 
    \item \textit{Tableau methods:} Interpolation proofs via tableaux for modal and intuitionistic logics~\cite{Fitting83}.
    \item \textit{Pitts' method:} Pitts' original paper on the proof of uniform interpolation for intuitionistic propositional logic~\cite{Pitts92}. Proof-theoretic methods for uniform interpolation for modal logics~\cite{Bilkova06}.
    \item \textit{Universal proof theory:} Semi-analytic rules and the existence of well-behaved sequent calculi via the examination of interpolation \cite{AkbarTabatabaiIJ21,AkbarTabatabaiIJ22int,AkbarTabatabaiIJ22nonnormal,AkbarTabatabaiJ18uniform,AkbarTabatabaiJ25,Iemhoff19AML,Iemhoff19APAL}. See also the published lecture notes in \cite{UPTLectureNotes} on which \Cref{sec:universal-proof-theory} is based.
    \item \textit{Interpolation via hypersequents:}  Craig/Lyndon interpolation for modal logics~\cite{Kuznets16Kolleg,Kuznets18}  and for intermediate logics~\cite{KuznetsL18}. Uniform modal interpolation~\cite{vdGiessenJK24}.
    \item \textit{Interpolation via nested sequents:} Craig/Lyndon interpolation for modal logics~\cite{FittingK15,Kuznets18} and for an intermediate logic~\cite{KuznetsL18} using linear nested sequents. Purely syntactic representation of the method~\cite{Lyon_etall20}. Uniform interpolation for modal logics~\cite{Bilkova11,vdGiessenJK24}.
    \item \textit{Interpolation  via labelled sequents:} Craig/Lyndon interpolation for modal logics~\cite{Kuznets16JELIA,Kuznets18}.
    \item \textit{Interpolation via custom-made generalizations of sequents:}  Hypersequents with an additional designated root component~\cite{vdGiessenJK23TABLEAUX}. Nested sequents with bounds on depth and/or branching~\cite{KuznetsL18,KuznetsL21}.
\end{itemize}
 
\subparagraph{Related topics} Here we would like to point to some related research directions on proof-theoretic methods of interpolation that fall beyond the scope of this chapter:
\begin{itemize}
    \item \textit{Interpolation as cut-introduction:} In~\cite{Saurin25}, a \textit{proof-relevant} view on Maehara's method in linear logic is provided in the sense that given a cut-free sequent proof $\pi$ of $\phi \seqar \psi$, there is a formula $\theta$ in the common signature of $\phi$ and $\psi$, and there are sequent proofs $\pi_1$, $\pi_2$ for $\phi \seqar \theta$ and $\theta \seqar \psi$, respectively, such that $\pi_1$ composed with $\pi_2$ cut-reduces to $\pi$. This leads to the introduction of cuts by synthesizing the interpolant.
    \item \textit{Interpolation via cyclic and non-wellfounded proofs:} \refchapter{chapter:fixedpoint} of this book discusses Maehara-type methods to construct interpolants for fixpoint logics via sequent proofs that may contain infinite branches. \
    \item \textit{Interpolation and proof complexity:} Interpolation has also been extensively studied in proof systems such as resolution and cutting planes [53, 80], where it plays a central role in proof complexity. In particular, \emph{feasible interpolation} investigates when interpolants can be efficiently extracted from proofs, thereby linking short proofs to small interpolating circuits. This connection provides a powerful bridge between proof complexity and circuit complexity, allowing lower bounds on circuit size to be transferred to lower bounds on proof length and enabling the classification of proof systems and the establishment of exponential separations. These themes are discussed in detail in \refchapter{chapter:proofcomplexity} of this book.
\end{itemize}

\addcontentsline{toc}{section}{Acknowledgments}
\section*{Acknowledgments}
The authors sincerely thank Patrick Koopmann, Rosalie Iemhoff, and Sam van Gool for careful reading of the earlier versions of this manuscript and for their insightful comments, as well  as Marianna Girlando for helpful discussions. Iris van der Giessen acknowledges the support by the Dutch Research Council (NWO) under the project \textit{Finding Interpolants: Proofs in Action} with file number VI.Veni.232.369 of the research program Veni. Roman Kuznets was supported by Czech Science Foundation Grant No.~22-06414L. 

\bibliography{bibliography}
\end{document}